\documentclass[oneside]{llncs}

\pdfoutput=1

\usepackage[utf8]{inputenc}
\usepackage[T1]{fontenc}
\usepackage{fullpage}
\usepackage{amsmath}
\usepackage{amssymb}
 \usepackage{amsthm}
 \usepackage{amsfonts}
\usepackage{physics}
\usepackage{marginnote}
\usepackage{dsfont}
\usepackage{mathtools}
\usepackage{url}
\usepackage{hyperref}
\hypersetup{colorlinks=true}
\usepackage{algorithm}
\usepackage{algorithmic}
\usepackage[table,xcdraw]{xcolor} \usepackage{adjustbox}
\usepackage{multirow}
\usepackage [operators,sets] {cryptocode}
\usepackage{bm}
\usepackage{enumitem}
\setlist{nosep} \setlist[1]{labelindent=\parindent}
\usepackage{qcircuit} 

\usepackage{csquotes}

\usepackage{tikz}
\usetikzlibrary{shapes,arrows,positioning}

\usepackage{booktabs} 

\usepackage{xfrac}

\makeatletter
\AtBeginDocument{\def\doi#1{\url{https://doi.org/#1}}}
\makeatother

\usepackage{caption}
\usepackage{subcaption}

\floatstyle{ruled}

\newfloat{protocol}{htb!}{idf}
\floatname{protocol}{Protocol}

\newfloat{resource}{htb!}{idf}
\floatname{resource}{Resource}

\newfloat{simulator}{htb!}{idf}
\floatname{simulator}{Simulator}

\newfloat{problem}{htb!}{idf}
\floatname{problem}{Problem}

\newcommand{\cptp}[1]{\mathsf{#1}}

\newcommand{\Id}{\cptp{I}}

\newcommand{\X}{\cptp{X}}
\newcommand{\Z}{\cptp{Z}}
\newcommand{\Y}{\cptp{Y}}

\newcommand{\RZ}{\cptp R_{\Z}}

\newcommand{\CZ}{\mathsf{CZ}}
\newcommand{\CNOT}{\mathsf{CNOT}}

\newcommand{\Abort}{\mathsf{Abort}}
\newcommand{\BQP}{\mathsf{BQP}}

\newcommand{\acc}{\mathsf{Acc}}

\newtheorem{result}{Result}

\let\ifcomment\iffalse 

\usepackage{marginnote}
\ifcomment
  \newcommand{\mn}[2]{\textsuperscript{\textcolor{red}{\textsf{\textbf{#1}}}}\marginnote{\tiny\textcolor{red} {\textsf{\textbf{#1:} #2}}}}
  \newcommand{\pn}[2]{\textcolor{red}{\textsf{\textbf{#1:} {#2}}}}
\else
  \newcommand{\mn}[2]{}
  \newcommand{\pn}[2]{}
\fi

\let\ifintro\iftrue  

 \usepackage{layout}
\begin{document}

\date{\today}
\title{Verification of Quantum Computations \newline without Trusted Preparations or Measurements}
\author{
Elham Kashefi\inst{1,2} \and
Dominik Leichtle\inst{2} \and
Luka Music\inst{3} \and
Harold Ollivier\inst{4}
}
\institute{
School of Informatics, University of Edinburgh, 10 Crichton Street, Edinburgh EH8 9AB, United Kingdom \and
Laboratoire d’Informatique de Paris 6, CNRS, Sorbonne Université, 4 Place Jussieu, 75005 Paris, France \and
Quandela, 7 Rue Léonard de Vinci, 91300 Massy, France \and
DI-ENS, Ecole Normale Supérieure, PSL Research University, CNRS, INRIA, 45 rue d'Ulm, 75005 Paris, France
}
 
\maketitle

\begin{abstract}
  With the advent of delegated quantum computing as a service, verifying quantum computations is becoming a question of great importance.  Existing information theoretically Secure Delegated Quantum Computing (SDQC) protocols require the client to possess the ability to perform either trusted state preparations or measurements. Whether it is possible to verify universal quantum computations with information-theoretic security without trusted preparations or measurements was an open question so far.  In this paper, we settle this question in the affirmative by presenting a modular, composable, and efficient way to turn known verification schemes into protocols that rely only on trusted gates.
  
  Our first contribution is an extremely lightweight reduction of the problem of quantum verification for $\BQP$ to the trusted application of single-qubit rotations around the $\Z$ axis and bit flips.
  The second construction presented in this work shows that it is generally possible to information-theoretically verify arbitrary quantum computations with quantum output without trusted preparations or measurements. However, this second protocol requires the verifier to perform multi-qubit gates on a register whose size is independent of the size of the delegated computation.
  
\end{abstract}

\setcounter{tocdepth}{3}

\section{Introduction}\label{sec:intro}

\subsection{Context and Motivation}

When quantum computers will have superseded classical machines for practical problems, ensuring that the results that they provide can be trusted will become a challenge. For some problems, the only strategies to verify the computation outcome after-the-fact are equivalent to performing the full classical simulation of the computation~\cite{BISB18characterizing,AABB19quantum}, thereby voiding the usefulness for quantum computers. This cannot be avoided, unless some widely believed complexity-theoretic conjecture -- such as the absence of collapse of the polynomial hierarchy -- is disproved.

Current trends in quantum computing point towards a future where most quantum computations are delegated to powerful quantum servers rather than executed by end-users with on-premise hardware. Delegation adds another layer of difficulty to the aforementioned testing challenge since it is not possible to observe \emph{in situ} the operations performed by the quantum computer. As a result, some of the assumptions which are reasonable in the non-delegated setting become hard to justify -- e.g.~the i.i.d. nature and non-maliciousness of the hardware operators actions. Therefore, tests which rely on these assumptions become immediately less convincing.

\begin{quote}
\itshape
What are then the available options for clients of QaaS companies which seek to ensure the integrity of their delegated computation?
\end{quote}

From a philosophical perspective, this question amounts to finding the minimal assumptions that are required to falsify quantum mechanics. The field of self-testing~\cite{MY03self} (see~\cite{SB20self} for a review) aims to answer this problem by relying on various Bell tests in which parties interact only classically with machines whose internal workings are unknown. If the tests pass, then depending on the specific scenario, it is possible to guarantee that the machine has done something which cannot be explained by local classical operations, and sometimes even that the system which the machines operate upon has a given quantum structure. These settings are highly adversarial and in the worst case no trust is placed in any of the machines. Unfortunately, most of the protocols resulting from these ideas have a very low throughput due to the necessity to enforce the absence of communication between some parties involved in the protocol, which in turn imposes a strict serialisation and the additional assumption of the space-like separation of the quantum devices. As such, state-of-the-art implementations are impractical even for the simplest task of key distribution~\cite{ZLRG22device}. Applying these techniques to testing delegated computations has been described~\cite{M16interactive} but the corresponding implementations would require quantum memories with extremely long storage times while throughput issues would further push back significantly the regime where quantum computers outperform classical ones.

This means that in any practically relevant settings, clients will need to make some assumptions if they want to verify their delegated computation. This translates into a trade-off between the requirements on the client and those on the server. On one end of the spectrum, clients can be fully classical if they are willing to rely of complexity assumptions such as the quantum hardness of LWE~\cite{M18classical}. In that case, they can use classical cryptographic schemes to run their computation of interest in a way that hides enough information to the server to make it possible to insert tests aimed at checking the honesty of the server alongside the computation. These schemes all require the server to implement complex unitary transformations on a quantum register whose size increases with the security parameter. This by itself requires thousands of error-corrected qubits and fault-tolerant gates, a feat which is out of reach of current and near-term hardware. On the other end, clients can operate and trust a small quantum device, either for generating or measuring single qubits in a small set of bases. It has been shown that this is sufficient to test arbitrary quantum computations with an exponentially-low probability to have the client accept a corrupted computation~\cite{FK17unconditionally,KW17optimised}. In the case of $\BQP$ computations, this can be done with no overhead on the server's side~\cite{LMKO21verifying}. Yet, the physical implementation still requires single-photon sources or detectors -- expensive and complex equipment pieces that need specialised teams to handle and service correctly.

In this paper, we show how to ensure the integrity of a delegated computation while reducing the need to trust complex devices on the client's side to the point where their physical implementations are both cheap and easy to maintain. This is done without degrading the cryptographic guarantees provided by the protocols nor drastically increasing the requirements on the server's side, both in terms of time or computation overhead.

\subsection{Our Contributions}

We focus our efforts on the client's operations required to implement Secure Delegated Quantum Computation protocols such as~\cite{FK17unconditionally,KW17optimised} belonging to the broader class of protocols described in~\cite{KKLM22unifying}. Their security relies on the ability of the client to implement a functionality called Remote State Preparation (RSP) that amounts to generating states from a given ensemble, e.g.~$\{(\ket{0} + e^{i k\pi/4}\ket 1)/\sqrt 2 \mid k = 0,\ldots 7\}$, on the server's side while not leaking additional information about which state has been prepared. This is usually performed by preparing those states on the client's side and using a quantum channel to transmit them to the server.

These protocols' security, as proven in the composable Abstract Cryptography (AC) framework~\cite{MR11abstract-cryptography,M12constructive-cryptography}, can be reduced to this specific subroutine. This allows us to focus on RSP independently of the rest of the protocol.
We show that RSP can be implemented securely without trusting the qubit preparation, if the client has access to trusted unitary transformations. The transformed states are sent to the server who performs the rest of the protocol with only classical interactions with the client.

Two aspects of such constructions make them particularly challenging: (i) the protocol replacing the trusted preparation must leak no information about the secretly prepared state to the server, and (ii) the protocol must not give the server the opportunity to influence the produced state in a way which depends on its classical description. As opposed to previous works~\cite{MKAC22qenclave} which were just aiming for \emph{blindness} and thus needed to fulfill only the property (i) above, property (ii) is necessary for the \emph{verifiability} of the resulting resource, and the harder of the two to satisfy. It requires that any deviation of the server during the protocol can be mapped to a deviation applied after the state has been generated, while maintaining its secret-independence. More precisely, for all possible malicious operations $\cptp E$ performed by the server before a trusted unitary $\cptp U_i$ is applied from the set $\{\cptp U_i\}_i$, the resulting deviation $\cptp U_i \cptp E \cptp U_i^\dagger$ when $\cptp E$ is commuted through $\cptp U_i$ should not depend on $i$ -- which is in general not guaranteed.

The contributions of this paper are twofold, expressed informally in Results~\ref{res:1} and \ref{res:2}.

\begin{result}[Secure Implementation of Single-Plane RSP without Trusted Preparation (Theorem~\ref{thm:rsp-from-rrd}]
  \label{res:1}
  It is possible to securely implement Remote State Preparations of ensembles of states in a single plane of the Bloch sphere using only trusted rotations in this plane and bit flips (Resource~\ref{res:rrd}).
\end{result}

Protocol~\ref{prot:rsp-from-rrd}, which yields the result above, works as follows. In the honest case, the client receives a $\ket{+}$ state from the server, applies a random bit-flip followed by a rotation by an angle $\theta$ around the $\Z$ axis of the Bloch sphere. It then sends back the state to the server. We show that this entirely replaces the trusted preparation of a $\Z(\theta)\ket{+}$ state even if the malicious server does not in fact send a $\ket{+}$ state. There is no degradation of the security error compared to using trusted preparations of states in the $\X - \Y$ plane.

Concerning the feasibility of this protocol, in the dual-rail encoding used on photonic chips, these operations can be implemented with a low number of simple components such as phase-shifters and beam-splitters. When combined with the tests from~\cite{KKLM23asymmetric} which only uses these kinds of states, we construct an efficient protocol for testing the correctness of $\BQP$ computations with only trusted rotations and bit-flips.

\begin{result}[Secure Implementation of Arbitrary RSP without Trusted Preparation (Theorem~\ref{thm:rsp-from-crsp-and-ru})]
  \label{res:2}
  It is possible to securely implement Remote State Preparations of arbitrary ensembles of states using trusted single-qubit unitaries and $\CNOT$ gates, whose number depends linearly in the security parameter and logarithmically in the size of the computation.
\end{result}
The construction is done in two steps. We first show in Protocol~\ref{prot:crsp-from-mrc} how to replace sending a Pauli eigenstate $\{\ket{0}, \ket{1}, \ket{+}, \ket{-}, \ket{+_i}, \ket{-_i}\}$ and then use this procedure in Protocol~\ref{prot:rsp-from-crsp-and-ru} to send arbitrary pure states.

In the honest case, the first step requires the server to first send to the client $n$ copies of the $\ket{0}$ state. The client then applies a random dephasing operator $\Z^d$ to each one and then a random invertible linear transformation $f : \bin^n \rightarrow \bin^n$. This mixes all the computational basis states apart from the all-zero state. The client then transforms the first qubit into its desired Pauli eigenstate via the appropriate single-qubit Clifford and encrypts it with a Quantum One-Time Pad. Random bit-flips are applied to the other qubits. All $n$ qubits are then sent back to the server who is then asked to measure the last $n - 1$ qubits and send the correct value for the bit-flips. We show that the server can do this with non-negligible probability if and only if it had effectively sent the all-zero state at the very beginning.

The second step takes as input the Pauli eigenstate $\ket{\psi} = \cptp C\ket{0}$ generated above, has the client apply a random single-qubit unitary $\cptp U$ and send this state back to the server. The client then sends the classical description of the unitary $\cptp V \cptp C^\dagger \cptp U^\dagger$ to the server, who applies it to obtain the client's desired state $\cptp V \ket{0}$. This second step is perfectly secure since the classical description of $\cptp V \cptp C^\dagger \cptp U^\dagger$ does not leak any information about the state $\cptp V \ket{0}$ thanks to the presence of the random unitary $\cptp C^\dagger \cptp U^\dagger$.

Combining these two steps allows us to construct an RSP for both rotated states $\Z(\theta)\ket{+}$ and computational basis states using only trusted transformations. This is the set of states that is required by most existing protocols for testing delegated computations such as the ones cited above. Replacing the trusted preparation step in any of these protocols, provided they are composably secure, by our construction yields a protocol in which the client only needs to trust a single application of a single-qubit unitary and a linear number of $\CNOT$ gates. The overall security error after this transformation degrades only negligibly compared to the protocol with trusted preparations so long as the client's desired computation is polynomial in the security parameter. While the size of the client's register grows with the security parameter, the server's overhead is significantly lowered when compared to computationally secure protocols with fully-classical clients. In addition, the size of the client's register is independent of the size of the computation.

Results~\ref{res:1} and~\ref{res:2} hold against unbounded adversarial servers. To our knowledge, this is the first result which removes the requirement of trusted preparation or measurement in the context of secure delegated computations while also preserving the same security characteristics -- i.e.~information-theoretic security and at most negligible degradation of the security error. Additionally, our protocols are composably secure, meaning that they can replace trusted preparations in any context beyond secure delegated quantum computations while preserving the same security, without needing to prove the security of the full protocol anew.

In addition to these two main contributions, we show in Appendix~\ref{app:non-verif} that if the goal is to construct a weaker version of RSP for arbitrary states -- 
namely one in which the server can deviate selectively on certain states and not on others 
-- then it is possible to do so with solely trusted single-qubit unitary transformations.

We also show in Appendix~\ref{app:collaborative_remote_operations} how any resource that is defined earlier in the paper can be lifted to the multi-party setting in which multiple clients want to generate a state from a certain ensemble on a server without leaking information to any malicious coalition. This allows us to generalise all our results to the realm of Delegated Quantum Secure Multi-Party Computations in which multiple clients want to jointly delegate a computation to an untrusted server without revealing their inputs or outputs to each other.

\subsection{Related Work}\label{subsec:related-work}

Single-qubit rotations in a single plane of the Bloch sphere have been shown to be sufficient if the goal is to perform blind delegated computations~\cite{MKAC22qenclave}, i.e.~if preventing information leaks about the computation, the inputs nor the outputs to the server is the only concern. They raise the question of whether it is also sufficient for verification, which we answer here for $\BQP$ computations.

Ref.~\cite{PLLC23multi} presents an experimental implementation of a multi-client blind delegated computation protocol in which a number of clients want to delegate a common computation to the server without revealing their inputs and outputs to each other and without revealing the computation to the server. They also allow the clients to use single-qubit rotation in a single plane of the Bloch sphere so our techniques could be used directly to upgrade their protocol and experiment to full security by adding a verification procedure at no additional hardware cost or change in the experiment's setup. The only modification would be the classical instructions sent to the server and classical post-processing performed on the client's side.

\section{Preliminaries}

\subsection{Abstract Cryptography}\label{sec:ac}

The Abstract Cryptography (AC) security framework \cite{MR11abstract-cryptography,M12constructive-cryptography} used in this work follows the \emph{ideal/real simulation paradigm}. A protocol is considered secure if it is a good approximation of an ideal version called a \emph{resource}.  The approximation should hold whether the participants follow the roles given to them by the protocol, if they are honest, or else deviate from these instructions arbitrarily. Protocols can use multiple such resources as building blocks to construct a new resource.

The main advantage is that any protocol that is proven secure in this framework is inherently \emph{composable}, in the sense that if two protocols are secure separately, the framework guarantees at an abstract level that their sequential or parallel execution is also secure. Intuitively, if a protocol A $\epsilon$-approximates a resource R and protocol B calls resource R to approximate another resource S, then replacing the call to R in B by protocol A does not degrade the security of the composed protocol by more than $\epsilon$. For this reason, Abstract Cryptography is considered one of the most stringent security framework, as stand-alone security criteria do not necessarily grant the security of composite protocols even when each constituent is secure.

We give more details about the Abstract Cryptography Framework in Appendix~\ref{app:ac}. We refer the reader to \cite{DFPR14composable} for a more in-depth presentation, in particular regarding the framework's composability in the context of SDQC.

\subsection{Secure Delegated Quantum Computation Protocols}
\label{sec:sdqc}

SDQC protocols allow a weak Client to delegate a computation to a powerful quantum Server while making sure that the server learns nothing about the computation --- the protocol is blind --- and is forced to either perform it correctly or force an abort --- the protocol is verified. This means that the probability that information about the computation leaks or that the client accepts an incorrect outcome are both vanishingly low. We present here the basic ideas for constructing a few such protocols, give their security results in Appendix \ref{app:ac} and defer to the respective original papers for more details.
See~\cite{GKK19verification} for an overview of previous quantum verification protocols.

\paragraph{Measurement-Based Quantum Computing.}
Most protocols for SDQC where clients have access to limited quantum resources are framed in the Measurement-Based Quantum Computation (MBQC) model.  A large entangled state is first constructed and then adaptive single-qubit measurement are performed on this state to perform the computation.

The most common resource states are so-called graph states. To each vertex $v$ of a graph $G$ is associated a qubit in the $\ket{+}$ state, while each edge $e$ corresponds to a $\CZ$ gate applied to the qubits associated to the vertices linked by $e$.

The single-qubit measurements are all in the $\X - \Y$ plane of the Bloch sphere, and an angle $\phi$ defines the measurement basis $\ket{\pm_\phi}$. Approximate universality is obtained for $\phi \in \Theta = \{k\pi/4\}_{k=0}^7$. The order in which measurements are performed is important since the measurement bases of future measurements need to be adapted based on a subset of previous outcomes. The order and corrections are given by a function called the g-flow $g$, whose existence for a given graph guarantees that the computation is independent of the measurement outcomes. See \cite{BKMP07generalized} for more details.

MBQC computations are therefore fully determined by the graph $G = (V, E)$, the measurement angles $\{\phi_v\}_{v \in V}$ and the g-flow $g$. 

\paragraph{Universal Blind Quantum Computing.}
It is possible to perfectly hide the measurement angles by having the Client prepare one $\ket{+_{\theta_v}}$ state for each vertex $v$, for a random $\theta_v \in \Theta$, send them to the Server who will perform the same entanglement operation according to the graph $G$. Then, the computation is performed by having the Client send the measurement angle $\delta_v := \phi'_v + \theta_v + r_v\pi$ to the Server, who performs the measurement and returns the outcome $b_v$. The Client recovers the correct outcome $s_v = b_v + r_v$. Intuitively, we can see that the angle $\theta_v$ perfectly hides the real computation angle $\phi'_v$ -- where the apostrophe denotes the corrected angle -- while $r_v$ perfectly hides the measurement outcome. The AC security of the Universal Blind Quantum Computation (UBQC) protocol, which implements this strategy, shows that the only information that the Server has access to is the structure of the graph \cite{BFK09universal,DFPR14composable}. This takes care of the blindness property.

\paragraph{Verifying Classical Input/Output Quantum Computations.}
If the computation is in the complexity class $\BQP$, then repeating the computation multiple times and taking the majority outcome will yield the correct outcome with overwhelming probability. We can then interleave test rounds based on the graph state's stabilisers together with these computations rounds and hide which is which with a UBQC overlay. The blindness forces the Server to attack indiscriminately a sizable number of test and computation rounds, if it wants to be able to change the Client's output. Therefore, if not more than a constant fraction of test rounds fail, then we can also guarantee with overwhelming probability that not too many of the computation rounds fail. In such cases, the majority vote will recover the correct outcome. This works so long as test rounds detect with constant probability all of the adversary's attacks which would be harmful to the computation \cite{KKLM22unifying}.

Up until recently, performing these test rounds required the Client to prepare not only the $\ket{+_\theta}$ states used by the UBQC protocol, but also computational basis states. Recently, a new protocol~\cite{KKLM23asymmetric} was proposed which does away with this additional requirement by aptly choosing the stabilisers used in the tests so that these contain no $\Z$ component. A test is performed by having the Client choose a random stabiliser which satisfies this property; and then have it instruct the corresponding measurements to the Server. The parity of measurement outcomes for the stabiliser determine the outcome of the test. Since these measurements are drawn from the same distribution as the regular computation thanks to the UBQC protocol, the Server cannot distinguish test rounds from computations.

Ref.~\cite{KKLM23asymmetric} shows the AC security of the SDQC protocol which uses this technique. More precisely, it shows that the probability that the Client accepts a corrupted outcome is negligibly small in $N$,  the total number of rounds in the protocol.

\paragraph{Verifying Quantum Input/Output Computation.}
If the Client's input to the computation is quantum, it might not be possible to repeat the computation. The no-cloning theorem forbids to perfectly copy the input which then impairs the ability to run the computation several times. Furthermore, if the output is quantum, it is unclear how to recombine the outputs of multiple executions as with the classical majority vote. One way to circumvent these issues is to perform the tests and computation in one single execution.

We present here the construction of \cite{KW17garbled}. Starting from a graph $G = (V, E)$, we first construct the Triple Graph $T(G) = (T(V), T(E))$ by replacing every vertex $v \in V$ by a set of three vertices $P_v = (v_1, v_2, v_3)$ and each edge $(v, w) \in E$ by the nine edges $(v_i, w_j)$ for $i, j \in \{1, 2, 3\}$. Then we construct the Dotted-Triple Graph $DT(G)$ by replacing each edge $e_{i, j} = (v_i, w_j) \in T(E)$ by a vertex $v_e$ and two edges $(v_i, v_{e_{i, j}})$ and $(v_{e_{i, j}}, w_i)$. These are called the added vertices.

To insert tests and the computation into $DT(G)$, the Client chooses in each set $P_v$ one computation, one dummy and one trap vertex at random. The type of the added vertices is chosen so that the two computation vertices are linked with another computation vertex and two dummies are linked by a trap. All other added vertices are dummies. The preparation now requires not only $\ket{+_\theta}$ states but also computational basis states. The latter will be used to prepare the dummies as they will then isolate in a blind way a subgraph that allows to carry on the desired computation and trap qubits that will serve to test the behavior of the server. The result is that the computation can be performed on a dotted version of the graph $G$, while there is always one trap among each set $P_v$ and one trap among the nine added vertices associated to each edge in $E$. The most efficient attack on a single position has a probability $8/9$ of being undetected.

The computation is encoded in a code which corrects up to $c$ errors so that the Server must corrupt $c$ computation positions before the output of the computation is modified. This has the same effect as repeating the computation in the previous case: the Server must attack a sizable number of positions while still having a constant probability of being detected for each position. Ref.~\cite{KW17garbled} shows that this results in a security error for the SDQC protocol which is negligible in $c$. The AC security of their protocol can be shown using the techniques from~\cite{DFPR14composable}.

\paragraph{Requirements for SDQC.}
To summarise, using the protocols described above:
\begin{itemize}
\item In order to perform their computation blindly and/or test $\BQP$ computations, the client must generate single qubits in the states $\{\ket{+_\theta}_\theta\}$.
\item In order to test arbitrary quantum computations with quantum inputs and outputs, the client must generate the states above and computational basis states.
\item The client needs to trust that its preparation method does not leak the state's description to the server and that the prepared states are correct -- up to operations that are independent of the secret parameter.
\end{itemize}

These three points are captured by the security theorems of the protocols described above in their respective papers. As mentioned above, we recall for completeness a simplified version in Appendix~\ref{subapp:sec-sdqc}. The final point is usually achieved by giving to the client the ability to prepare and send these states using a small trusted quantum device that they control in a secure environment. This is clearly the strongest assumption from the client's perspective as producing and preparing a quantum state require complex devices. We now show how to replace this requirement with the requirement to securely apply a restricted class of operations to quantum systems that are otherwise produced by the untrusted server. By reducing the complexity of the devices to trust this entails stronger security.

\section{Single-Plane RSP from Remote Rotation}\label{sec:rsp-from-rrd}

The following Resource models in the AC framework the quantum capabilities of the weak Client in the SDQC protocol described in the previous section. It simply prepares a state in the $\X - \Y$ plane and sends it out. The AC security of the SDQC protocol from \cite{KKLM23asymmetric} reduces to the Client's capability of implementing this resource.

\begin{resource}[H]
  \caption{Single-Plane Remote State Preparation (SP-RSP)}
  \label{res:sp-rsp}
  \begin{algorithmic}[0]
    \STATE \textbf{Inputs:} The Sender inputs an angle $\theta \in \Theta = \qty{\frac{k\pi}{4}}_{k \in \qty{0, \ldots, 7}}$.
    \STATE \textbf{Computation by the Resource:} The Resource prepares and sends a qubit in state $\ket{+_\theta}$ to the Receiver.
  \end{algorithmic}
\end{resource}

\begin{figure}[t]
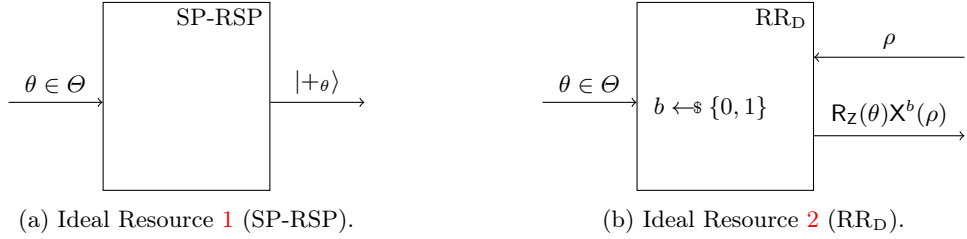

\centering
\begin{subfigure}[t]{0.45\textwidth}
  \centering
  \begin{bbrenv}{sprsp}
    \begin{bbrbox}[name=SP-RSP, minheight=2.5cm]\end{bbrbox}
    \bbrmsgspace{12mm}
    \bbrmsgto{top={$\theta \in \Theta$}}
    \bbrqryspace{12mm}
    \bbrqryto{top={$\ket{+_\theta}$}}
  \end{bbrenv}
  \caption{Ideal Resource~\ref{res:sp-rsp} (SP-RSP).}
  \label{fig:sp-rsp}
\end{subfigure}
\begin{subfigure}[t]{0.45\textwidth}
  \centering
  \begin{bbrenv}{rrd}
    \begin{bbrbox}[name=$\text{RR}_\text{D}$, minheight=2.5cm]\pseudocode{b \sample \bin
      }\end{bbrbox}
    \bbrmsgspace{12mm}
    \bbrmsgto{top={$\theta \in \Theta$}}
    \bbrqryspace{6mm}
    \bbrqryfrom{top={$\rho$}, length=2cm}
    \bbrqryspace{5mm}
    \bbrqryto{top={$\RZ(\theta) \X^b(\rho)$}, length=2cm}
  \end{bbrenv}
  \caption{Ideal Resource~\ref{res:rrd} ($\text{RR}_\text{D}$).}
  \label{fig:rrd}
\end{subfigure}
\caption{Ideal functionalities for Single-Plane Remote State Preparation (SP-RSP), and Remote Rotation with Dephasing ($\text{RR}_\text{D}$).}
\label{fig:sp-rsp_rrd}
\end{figure}

\begin{remark}\label{remark:equivalence-of-single-plane}
  Resource~\ref{res:sp-rsp} fixes the remotely prepared states to be in the $\X - \Y$-plane. Note however that the choice of plane is arbitrary in the sense that the choice of any other plane would yield an equivalent resource in the AC framework. To transform one such resource into another it suffices for the receiver to perform a change of basis after receiving the state from the available SP-RSP resource.
\end{remark}

We now introduce a new Resource, in which a party receives a single-qubit state, applies to it a random bit-flip and a rotation around the $\Z$ axis and sends it back. For a CPTP map $\cptp C$ acting on $n$ qubits and an $n$-qubit mixed state $\rho$, we will use the notation $\cptp C(\rho)$ to mean $\cptp C \rho \cptp C^\dagger$.

\begin{resource}[ht]
  \caption{Remote Rotation with Dephasing ($\text{RR}_\text{D}$)}
  \label{res:rrd}
  \begin{algorithmic}[0]
    \STATE \textbf{Inputs:}
    \begin{itemize}
    \item The Sender inputs an angle $\theta \in \Theta$.
    \item The Receiver inputs a single qubit in state $\rho$.
    \end{itemize}
    \STATE \textbf{Computation by the Resource:}
    \begin{enumerate}
    \item The Resource samples a bit $b \sample \bin$ uniformly at random.
    \item The Resource outputs a qubit in state $\RZ(\theta) \X^b(\rho)$ to the Receiver.
    \end{enumerate}
  \end{algorithmic}
\end{resource}

\begin{remark}
  Remark~\ref{remark:equivalence-of-single-plane} analogously applies to the plane chosen in Resource~\ref{res:rrd}: to apply RR in a different plane, the Receiver performs a unitary change of basis both before and after using Resource~\ref{res:rrd}.
\end{remark}

The $\text{RR}_\text{D}$ Resource gives less power to the Sender compared to the SP-RSR Resource since it does not control the state which the Receiver inputs. However, we demonstrate that the SP-RSP Resource can in fact be constructed with only one call to the $\text{RR}_\text{D}$ Resource via the following Protocol~\ref{prot:rsp-from-rrd}. The Receiver simply produces a $\ket{+}$ state, sends it to the Sender, who applies a random bit-flip and a rotation and sends it back.

\begin{protocol}[ht]
  \caption{SP-RSP from $\text{RR}_\text{D}$}
  \label{prot:rsp-from-rrd}
  \begin{algorithmic} [0]
    \STATE \textbf{Input:} The Sender inputs an angle $\theta \in \Theta$.
    \STATE \textbf{Protocol:}
    \begin{enumerate}
    \item The Sender and Receiver call the $\text{RR}_\text{D}$ Resource~\ref{res:rrd}:
    \begin{itemize}
    	\item The Sender inputs the angle $\theta$.
    	\item The Receiver inputs a qubit in the state $\ket{+}$.
    	\item The Resource returns a single qubit to the Receiver who sets it as its output.
    \end{itemize}
    \end{enumerate}
  \end{algorithmic}
\end{protocol}

The security of Protocol~\ref{prot:rsp-from-rrd} is given in the following Theorem~\ref{thm:rsp-from-rrd}, whose proof can be found in Appendix~\ref{app:rsp-from-rrd}.

\begin{theorem}[Security of Protocol \ref{prot:rsp-from-rrd}]
\label{thm:rsp-from-rrd}
\newcounter{count:rsp-from-rrd}
\setcounterref{count:rsp-from-rrd}{thm:rsp-from-rrd}
  Protocol~\ref{prot:rsp-from-rrd} perfectly constructs Resource~\ref{res:sp-rsp} (SP-RSP) from Resource~\ref{res:rrd} ($\text{RR}_\text{D}$).
\end{theorem}

We then use the composability of the AC framework along with Theorems~\ref{thm:rsp-from-rrd} and~\ref{thm:sp-sdqc} to construct an SDQC protocol for $\BQP$ computations without trusted preparations or measurements. The security error of the resulting protocol is the same as in Theorem~\ref{thm:sp-sdqc} since Protocol~\ref{prot:rsp-from-rrd} incurs no security loss.

\begin{corollary}[SDQC for $\BQP$ From Trusted Rotations]
\label{cor:sdqc-bqp}
There exists an efficient protocol which $\epsilon$-constructs the SDQC Resource \ref{res:sdqc} for $\BQP$ computations from $\abs{V}N$ instances of the Remote Rotation with Dephasing Resource~\ref{res:rrd}, while only leaking the graph $G = (V, E)$ used in the MBQC computation, with $\epsilon$ negligibly small in $N$.
\end{corollary}

\section{RSP of Arbitrary Single-qubit States from Remote Clifford and Unitary}\label{sec:arbitrary-states}

In this section, we show how to prepare arbitrary single-qubit states without using trusted preparations or measurements. The price to pay for this increase in functionality is an increase in the Sender's computational power. Instead of single-qubit operations as in Section~\ref{sec:rsp-from-rrd}, the Sender must now perform Clifford operations on states whose size depends on the security parameter. Since we show that the security error of the protocol decreases exponentially in the security parameter, the protocol remains resource-efficient.

\begin{resource}[ht]
  \caption{Clifford Remote State Preparation (C-RSP)}
  \label{res:c-rsp}
  \begin{algorithmic}[0]
    \STATE \textbf{Inputs:} The Sender inputs the classical description of a single-qubit Clifford operator $\cptp C$.
    \STATE \textbf{Computation by the Resource:} The Resource prepares and sends a qubit in state $\cptp C\ket{0}$ to the Receiver.
  \end{algorithmic}
\end{resource}

Resource~\ref{res:c-rsp} (C-RSP) is a first step towards preparing more general states since it can generate all six single-qubit Pauli eigenstates. Then, Resource~\ref{res:rsp} (RSP) is a generalisation of both Resources~\ref{res:sp-rsp} and \ref{res:c-rsp} which allows for the remote preparation of arbitrary single-qubit states.

\begin{resource}[ht]
  \caption{Remote State Preparation (RSP)}
  \label{res:rsp}
  \begin{algorithmic}[0]
    \STATE \textbf{Inputs:} The Sender inputs the classical description of a single-qubit unitary $\cptp U$.
    \STATE \textbf{Computation by the Resource:} The Resource prepares and sends a qubit in state $\cptp U\ket{0}$ to the Receiver.
  \end{algorithmic}
\end{resource}

\begin{figure}[t]
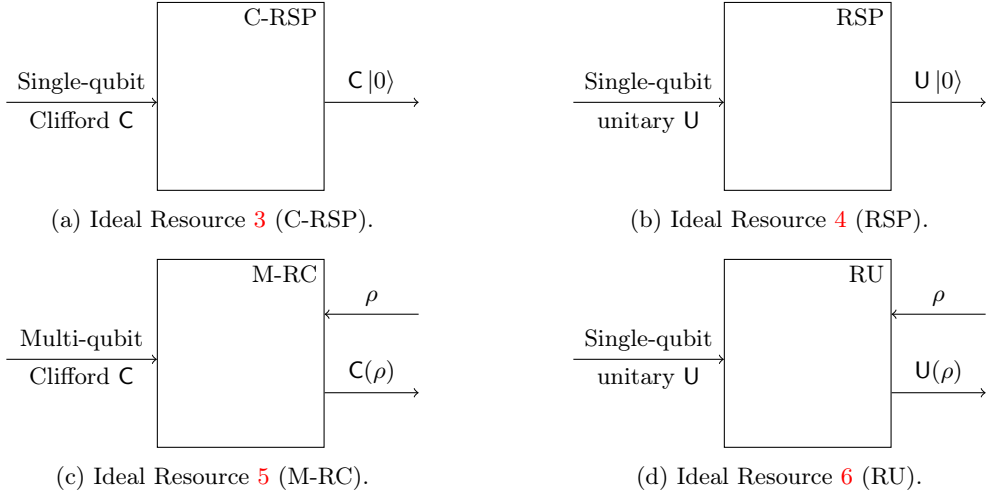

\centering
\begin{subfigure}[t]{0.45\textwidth}
  \centering
  \begin{bbrenv}{crsp}
    \begin{bbrbox}[name=C-RSP, minheight=2.5cm]\end{bbrbox}
    \bbrmsgspace{12mm}
    \bbrmsgto{top={Single-qubit}, bottom={Clifford $\cptp C$}, length=2cm}
    \bbrqryspace{12mm}
    \bbrqryto{top={$\cptp C\ket{0}$}}
  \end{bbrenv}
  \caption{Ideal Resource~\ref{res:c-rsp} (C-RSP).}
  \label{fig:c-rsp}
\end{subfigure}
\begin{subfigure}[t]{0.45\textwidth}
  \centering
  \begin{bbrenv}{rsp}
    \begin{bbrbox}[name=RSP, minheight=2.5cm]\end{bbrbox}
    \bbrmsgspace{12mm}
    \bbrmsgto{top={Single-qubit}, bottom={unitary $\cptp U$}, length=2cm}
    \bbrqryspace{12mm}
    \bbrqryto{top={$\cptp U\ket{0}$}}
  \end{bbrenv}
  \caption{Ideal Resource~\ref{res:rsp} (RSP).}
  \label{fig:rsp}
\end{subfigure}
\\[2ex]
\begin{subfigure}[t]{0.45\textwidth}
  \centering
  \begin{bbrenv}{mrc}
    \begin{bbrbox}[name=M-RC, minheight=2.5cm]\end{bbrbox}
    \bbrmsgspace{12mm}
    \bbrmsgto{top={Multi-qubit}, bottom={Clifford $\cptp C$}, length=2cm}
    \bbrqryspace{6mm}
    \bbrqryfrom{top={$\rho$}}
    \bbrqryspace{5mm}
    \bbrqryto{top={$\cptp C(\rho)$}}
  \end{bbrenv}
  \caption{Ideal Resource~\ref{res:m-rc} (M-RC).}
  \label{fig:m-rc}
\end{subfigure}
\begin{subfigure}[t]{0.45\textwidth}
  \centering
  \begin{bbrenv}{ru}
    \begin{bbrbox}[name=RU, minheight=2.5cm]\end{bbrbox}
    \bbrmsgspace{12mm}
    \bbrmsgto{top={Single-qubit}, bottom={unitary $\cptp U$}, length=2cm}
    \bbrqryspace{6mm}
    \bbrqryfrom{top={$\rho$}}
    \bbrqryspace{5mm}
    \bbrqryto{top={$\cptp U(\rho)$}}
  \end{bbrenv}
  \caption{Ideal Resource~\ref{res:ru} (RU).}
  \label{fig:ru}
\end{subfigure}
\caption{Ideal functionalities for Clifford Remote State Preparation (C-RSP), Remote State Preparation (RSP), Multi-qubit Remote Clifford (M-RC), and Remote Unitary (RU).}
\label{fig:c-rsp_rsp_m-rc}
\end{figure}

To construct Resource~\ref{res:rsp}, we start by showing how to construct Resource~\ref{res:c-rsp} without trusted preparations or measurements. This will allows us to \emph{break out} of the single-plane prison. This first step requires the Sender to actively test the Receiver for possible deviations, and hence incurs a (controllable) security loss, whereas the protocol which construct arbitrary RSP from Clifford RSP will be perfectly secure. The Sender's capabilities are now captured by Resource~\ref{res:m-rc}: it can apply any Clifford of its choice to $n$ qubits.

\begin{resource}[ht]
  \caption{Multi-qubit Remote Clifford (M-RC)}
  \label{res:m-rc}
  \begin{algorithmic}[0]
    \STATE \textbf{Inputs:}
    \begin{itemize}
    \item The Sender inputs the classical description of an $n$-qubit Clifford operator $\cptp C$.
    \item The Receiver inputs $n$ qubits in state $\rho$.
    \end{itemize}
    \STATE \textbf{Computation by the Resource:} The Resource outputs $n$ qubits in state $\cptp C(\rho)$ to the Receiver.
  \end{algorithmic}
\end{resource}

The goal of Protocol~\ref{prot:crsp-from-mrc} if to produce the state $\cptp C \ket{0}$ for a Clifford $\cptp C$ chosen by the Sender. The Receiver starts by preparing the state $\ket{\mathbf{0}}_n$, where the index denote the size. It sends it to the Sender, who applies a dephasing $\Z^{\mathbf{d}}$ and a unitary $\cptp U_g$ for a random invertible linear function $g$ over $\mathbb{Z}_2^n$.\footnote{$g$ is chosen by sampling uniformly at random a basis $(\mathbf{b}_1, \ldots, \mathbf{b}_n)$ of $\mathbb{Z}_2^n$ and setting $g(\mathbf{x}) = x_1\mathbf{b}_1 \oplus \ldots x_n\mathbf{b}_n$, where $x_i$ is the $i$\textsuperscript{th} the bits of $\mathbf{x}$ and $\oplus$ is the addition operation of $\mathbb{Z}_2^n$.} For $\mathbf{x} \in \mathbb{Z}_2^n$, we have $U_g\ket{\mathbf{x}} = \ket{g(\mathbf{x})}$. Then, this unitary preserves the state $\ket{\mathbf{0}}_n$ but perfectly randomises all other computational basis states.\footnote{Since $g(\mathbf{0}) = \mathbf{0}$, then $U_g\ket{\mathbf{0}}_n = \ket{\mathbf{0}}_n$.} The Sender then applies its desired single-qubit Clifford $\cptp C$ to the first qubit and encrypts it with a Quantum One-Time-Pad $\cptp P_2$. It applies $\X^{\mathbf{r}}$ to the other $n-1$ qubits before sending back the $n$-qubit state. The Receiver has to measure the last $n-1$ qubits and send the outcome $\mathbf{r}'$ to the Sender, who checks that it matches $\mathbf{r}$. If the test passes, the Sender sends the decryption key $\cptp P_2^\dagger$, otherwise it aborts. The Receiver can then decrypt the state and get the correct outcome.

\begin{protocol}[ht]
  \caption{C-RSP from M-RC}
  \label{prot:crsp-from-mrc}
  \begin{algorithmic} [0]
  	\STATE \textbf{Public Information:} Security parameter $n$.
    \STATE \textbf{Input:} The Sender inputs the classical description of a single-qubit Clifford operator $\cptp C$.
    \STATE \textbf{Protocol:}
    \begin{enumerate}
    \item The Sender samples uniformly at random bit-strings $\mathbf{d} \sample \bin^n$ and $\mathbf{r} \sample \bin^{n-1}$. It samples uniformly at random a single-qubit Pauli operator $\cptp P_2 \sample \mathcal{P}_1$. It samples a random invertible linear function $g$ over $\mathbb{Z}_2^n$. The Sender computes the $n$ qubit Clifford:
	\begin{align*}
    \cptp C_1 = \left(\cptp P_2\cptp C\otimes \X^{\mathbf{r}}\right)\cptp U_g\Z^{\mathbf{d}},
    \end{align*}	
    where $\cptp P_2, \cptp C$ act on the first qubit, $\X^{\mathbf{r}}$ acts on the other $n-1$ qubits and $\cptp U_g\ket{\mathbf{s}} = \ket{g(\mathbf{s})}$ for all $\mathbf{s} \in \bin^n$.
    \item The Sender and Receiver call the M-RC Resource~\ref{res:m-rc}:
    \begin{itemize}
    	\item The Sender inputs the classical description of $\cptp C_1$.
    	\item The Receiver inputs $n$ qubits in the state $\ket{\mathbf{0}}_n$.
    	\item The Receiver receives $n$ qubits as output.
    \end{itemize}
    \item The Receiver measures the last $n-1$ qubits in the computational basis and sends the measurement outcome $\mathbf{r}' \in \bin^{n-1}$ to the Sender.
    \item If $\mathbf{r}' \neq \mathbf{r}$, the Sender sends $\Abort$. Otherwise it sends the classical description of $\cptp P_2^\dagger$ to the Receiver.
    \item The Receiver applies $\cptp P_2^\dagger$ to the unmeasured qubit and sets it as its output.
    \end{enumerate}
  \end{algorithmic}
\end{protocol}

Intuitively, if the Receiver is honest then both the dephasing and $\cptp U_g$ have no effect and $\mathbf{r} = \mathbf{r}'$, otherwise $\mathbf{r}'$ is completely random. This is captured by Theorem~\ref{thm:crsp-from-mrc} below, whose proof is given in Appendix~\ref{app:crsp-from-mrc}.

\begin{theorem}[Security of Protocol \ref{prot:crsp-from-mrc}]
\label{thm:crsp-from-mrc}
\newcounter{count:crsp-from-mrc}
\setcounterref{count:crsp-from-mrc}{thm:crsp-from-mrc}
  Let $n$ be the security parameter used in Protocol~\ref{prot:crsp-from-mrc}. Then Protocol~\ref{prot:crsp-from-mrc} $\epsilon_n$-constructs Resource~\ref{res:c-rsp} (C-RSP) from Resource~\ref{res:m-rc} (M-RC), for $\epsilon_n = \frac{1}{2^n - 1}$.
\end{theorem}

In order to recover the full Bloch sphere we introduce the Remote Unitary Resource which allows the Sender to apply a unitary of its choice to a state before sending it back to the Receiver.

\begin{resource}[ht]
  \caption{Remote Unitary (RU)}
  \label{res:ru}
  \begin{algorithmic}[0]
    \STATE \textbf{Inputs:}
    \begin{itemize}
    \item The Sender inputs the classical description of a single-qubit unitary $\cptp U$.
    \item The Receiver inputs a single qubit in state $\rho$.
    \end{itemize}
    \STATE \textbf{Computation by the Resource:} The Resource outputs a qubit in state $\cptp U(\rho)$ to the Receiver.
  \end{algorithmic}
\end{resource}

We can then combine the C-RSP Resource with the RU Resource to construct a protocol for RSP with arbitrary states. The protocol simply sends a random Clifford state $\cptp C\ket{0}$ to the Receiver using the C-RSP Resource, who then inputs it into the RU Resource. This Resource is used to apply a random unitary $\cptp U_1$ to the state. The Sender can then send the classical description of unitary $\cptp U_2 = \cptp U \cptp C^\dagger \cptp U_1^\dagger$ for the Receiver to apply to recover the correct state. Since $\cptp C^\dagger \cptp U_1^\dagger$ is perfectly random, it randomises $\cptp U_2$ which therefore leaks no information about the desired state $\cptp U \ket{0}$.

\begin{protocol}[ht]
  \caption{RSP from C-RSP and RU}
  \label{prot:rsp-from-crsp-and-ru}
  \begin{algorithmic} [0]
    \STATE \textbf{Input:} The Sender inputs the classical description of a single-qubit unitary $\cptp U$.
    \STATE \textbf{Protocol:}
    \begin{enumerate}
    \item The Sender samples a single-qubit Clifford operator $\cptp C$ uniformly at random and sends its classical description to Resource~\ref{res:c-rsp} (C-RSP), which sends a qubit in the state $\cptp C\ket{0}$ to the Receiver.
    \item The Sender and Receiver call the RU Resource~\ref{res:ru}:
    \begin{itemize}
    	\item The Sender inputs the classical description of a Haar-random single-qubit unitary $\cptp U_1$.
    	\item The Receiver inputs the qubit that it received in step $1$ from the C-RSP Resource.
    	\item The Receiver receives a single qubit as output.
    \end{itemize}
    \item The Sender computes $\cptp U_2 = \cptp U \cptp C^\dagger \cptp U_1^\dagger$, and sends its classical description to the Receiver.
    \item The Receiver applies $\cptp U_2$ to the qubit received in step $2$ from the RU Resource and sets it as its output.
    \end{enumerate}
  \end{algorithmic}
\end{protocol}

The security of Protocol \ref{prot:rsp-from-crsp-and-ru} is given below, with the proof being in Appendix~\ref{app:rsp-from-crsp-and-ru}.

\begin{theorem}[Security of Protocol~\ref{prot:rsp-from-crsp-and-ru}]
\label{thm:rsp-from-crsp-and-ru}
\newcounter{count:rsp-from-crsp-and-ru}
\setcounterref{count:rsp-from-crsp-and-ru}{thm:rsp-from-crsp-and-ru}
  Protocol~\ref{prot:rsp-from-crsp-and-ru} perfectly constructs Resource~\ref{res:rsp} (RSP) from Resource~\ref{res:c-rsp} (C-RSP) and Resource~\ref{res:ru} (RU).
\end{theorem}

\begin{remark}
Note that it is possible to construct directly the RSP Resource for arbitrary states if we simply replace $\cptp C$ in Protocol~\ref{prot:crsp-from-mrc} with the Sender's unitary $\cptp U$. The proofs follow as before with no change whatsoever. We choose to present it in this way to preserve the modularity of the construction.
\end{remark}

\begin{remark}
It is possible to construct an RSP resource for arbitrary single-qubit states from the Remote Unitary Resource alone if we are willing to give more power to the adversary. We demonstrate this in Appendix~\ref{app:non-verif} by constructing the RSP with Selective NOT Resource (RSP-SN). This resource still prepares a state in a basis of the Sender's choice, but the adversarial Receiver can choose which one of the basis' states is produced. In the end, the received state is of the form $\cptp U \ket{b}$ for a unitary $\cptp U$ chosen by the Sender and a bit $b$ chosen by the Receiver.

This RSP-SN Resource allows the Receiver to selectively deviate depending on the state chosen by the Sender. For example, if the Sender prepares either $\ket{+_\theta}$ or $\ket{b}$, then the Receiver can always flip the output by setting $b = 1$ and later apply $\Z$. This second operation will flip back the rotated states but not the computational basis states. None of the SDQC protocols that exist as of now take into account such selective deviations -- and there are known attacks which leverage this effect, such as the one presented in \cite{KKLM23asymmetric}. Whether it is possible to design an SDQC protocol that can use the RSP-SN Resource as a building block by mitigating the additional power given to the adversary is therefore an interesting open question. 
\end{remark}

We again use the composability of the AC framework along with Theorems~\ref{thm:crsp-from-mrc}, \ref{thm:rsp-from-crsp-and-ru} and~\ref{thm:full-sdqc} to construct an SDQC protocol for arbitrary quantum computations without trusted preparations or measurements. The security error of the resulting protocol is increased by $\frac{1}{2^n - 1}$ compared to Theorem~\ref{thm:full-sdqc} every time the RSP Resource is replaced by the compositions of Protocols~\ref{prot:crsp-from-mrc} and \ref{prot:rsp-from-crsp-and-ru}. However, since the number of calls to the RSP Resource in the initial protocol is linear in the size of the computation and security parameter, the overall loss is still negligible in the security parameter.

\begin{corollary}[SDQC for Arbitrary Quantum Computations From Trusted Unitary Operations]
There exists an efficient protocol which $\epsilon'$-constructs the SDQC Resource \ref{res:sdqc} for arbitrary quantum computations from $\order{c(\abs{V} + \abs{E})}$ instances of the $n$-qubit Remote Clifford Resource~\ref{res:m-rc} and Remote Unitary Resource~\ref{res:ru}, while only leaking the graph $G = (V, E)$ used in the MBQC computation, with $\epsilon'$ negligibly small in $c$ and $n$. The Remote Unitary only needs to be able to apply $\Z$-rotations, Hadamard, and Pauli $\X$ operations.
\end{corollary}

\section{Discussion}\label{sec:discussion}

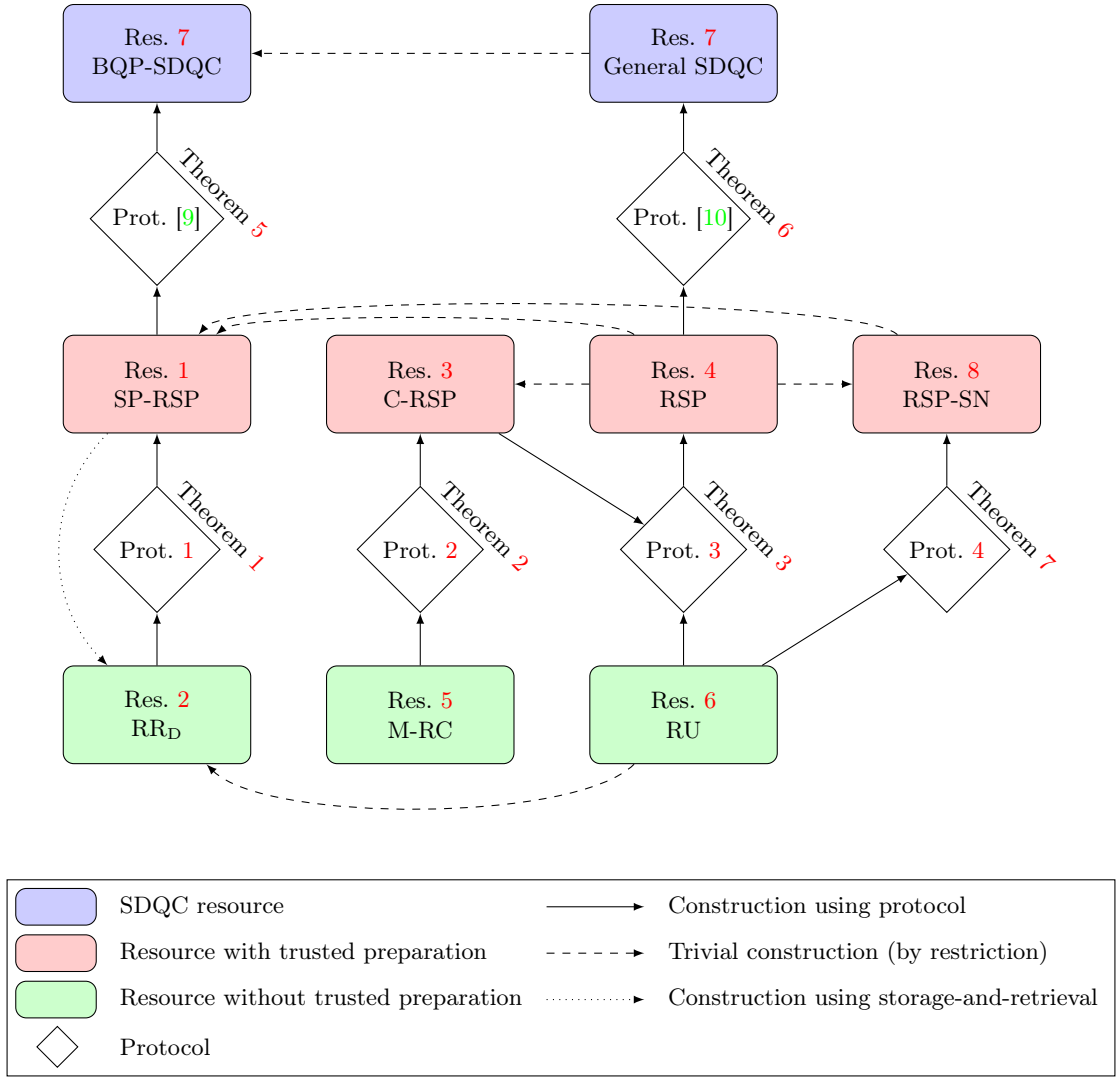
\begin{figure}[htp]
\centering
\begin{tikzpicture}[node distance=2.2cm, auto]
\tikzstyle{noprepresource} = [rectangle, draw, fill=green!20, text width=7em, text centered, rounded corners, minimum height=4em]
	\tikzstyle{prepresource} = [rectangle, draw, fill=red!20, text width=7em, text centered, rounded corners, minimum height=4em]
	\tikzstyle{appliedresource} = [rectangle, draw, fill=blue!20, text width=7em, text centered, rounded corners, minimum height=4em]
	\tikzstyle{protocol} = [diamond, draw, text width=4.5em, text badly centered, inner sep=0pt]
	\tikzstyle{line} = [draw, -]
	\tikzstyle{arrow} = [draw, -latex]
	\tikzstyle{dashedarrow} = [draw, -latex, dashed]
	
\node [appliedresource] (res-bqp-sdqc) {Res.~\ref{res:sdqc} \\ BQP-SDQC};
    \node [protocol, below of=res-bqp-sdqc,label={[rotate=-45]above right:{\hspace*{-5mm}Theorem~\ref{thm:sp-sdqc}}}] (prot-bqp-sdqc) {Prot.~\cite{KKLM23asymmetric}};
    \node [prepresource, below of=prot-bqp-sdqc] (res-sp-rsp) {Res.~\ref{res:sp-rsp} \\ SP-RSP};
    \node [prepresource, right of=res-sp-rsp, node distance=3.5cm] (res-c-rsp) {Res.~\ref{res:c-rsp} \\ C-RSP};
    \node [prepresource, right of=res-c-rsp, node distance=3.5cm] (res-rsp) {Res.~\ref{res:rsp} \\ RSP};
    \node [protocol, above of=res-rsp,label={[rotate=-45]above right:{\hspace*{-5mm}Theorem~\ref{thm:full-sdqc}}}] (prot-general-sdqc) {Prot.~\cite{KW17garbled}};
    \node [appliedresource, above of=prot-general-sdqc] (res-general-sdqc) {Res.~\ref{res:sdqc} \\ General SDQC};
    \node [prepresource, right of=res-rsp, node distance=3.5cm] (res-rsp-sn) {Res.~\ref{res:rsp-sn} \\ RSP-SN};
    \node [protocol, below of=res-sp-rsp,label={[rotate=-45]above right:{\hspace*{-5mm}Theorem~\ref{thm:rsp-from-rrd}}}] (prot-sp-rsp) {Prot.~\ref{prot:rsp-from-rrd}};
    \node [protocol, below of=res-c-rsp,label={[rotate=-45]above right:{\hspace*{-5mm}Theorem~\ref{thm:crsp-from-mrc}}}] (prot-c-rsp) {Prot.~\ref{prot:crsp-from-mrc}};
    \node [protocol, below of=res-rsp,label={[rotate=-45]above right:{\hspace*{-5mm}Theorem~\ref{thm:rsp-from-crsp-and-ru}}}] (prot-rsp) {Prot.~\ref{prot:rsp-from-crsp-and-ru}};
    \node [protocol, below of=res-rsp-sn,label={[rotate=-45]above right:{\hspace*{-5mm}Theorem~\ref{thm:rsp-sn-from-ru}}}] (prot-rsp-sn) {Prot.~\ref{prot:rsp-sn-from-ru}};
    \node [noprepresource, below of=prot-sp-rsp] (res-rrd) {Res.~\ref{res:rrd} \\ $\text{RR}_\text{D}$};
    \node [noprepresource, right of=res-rrd, node distance=3.5cm] (res-m-rc) {Res.~\ref{res:m-rc} \\ M-RC};
    \node [noprepresource, right of=res-m-rc, node distance=3.5cm] (res-ru) {Res.~\ref{res:ru} \\ RU};
    
\path [arrow] (res-sp-rsp) -- (prot-bqp-sdqc);
    \path [arrow] (prot-bqp-sdqc) -- (res-bqp-sdqc);
    \path [arrow] (res-rsp) -- (prot-general-sdqc);
    \path [arrow] (prot-general-sdqc) -- (res-general-sdqc);
    \path [arrow] (res-rrd) -- (prot-sp-rsp);
    \path [arrow] (prot-sp-rsp) -- (res-sp-rsp);
    \path [arrow] (res-m-rc) -- (prot-c-rsp);
    \path [arrow] (prot-c-rsp) -- (res-c-rsp);
    \path [arrow] (res-c-rsp) -- (prot-rsp);
    \path [arrow] (res-ru) -- (prot-rsp);
    \path [arrow] (prot-rsp) -- (res-rsp);
    \path [arrow] (res-ru) -- (prot-rsp-sn);
    \path [arrow] (prot-rsp-sn) -- (res-rsp-sn);
\path [dashedarrow] (res-general-sdqc) -- (res-bqp-sdqc);
    \path [dashedarrow] (res-rsp) -- (res-c-rsp);
	\path [dashedarrow] (res-rsp) -- (res-rsp-sn);
	\draw (res-rsp) edge[out=135,in=40,looseness=0.2,-latex,dashed] (res-sp-rsp);
	\draw (res-rsp-sn) edge[out=135,in=50,looseness=0.2,-latex,dashed] (res-sp-rsp);
	\draw (res-ru) edge[out=225,in=315,looseness=0.5,-latex,dashed] (res-rrd);
\draw (res-sp-rsp) edge[out=225,in=135,looseness=1,-latex,dotted] (res-rrd);

\matrix [column sep=2mm, row sep=1mm, draw, anchor=west, xshift=1cm, below of=res-rrd, node distance=3.5cm, anchor=west, xshift=-3cm]
	{
		\node [appliedresource, text width=2.625em, minimum height=1.5em] {}; & \node {SDQC resource}; & \path [arrow] (0, 0) -- (1.3, 0); & \node {Construction using protocol}; \\
		\node [prepresource, text width=2.625em, minimum height=1.5em] {}; & \node {Resource with trusted preparation}; & \path [dashedarrow] (0, 0) -- (1.3, 0); & \node {Trivial construction (by restriction)}; \\
		\node [noprepresource, text width=2.625em, minimum height=1.5em] {}; & \node {Resource without trusted preparation}; & \path [arrow,dotted] (0, 0) -- (1.3, 0); & \node {Construction using storage-and-retrieval}; \\
		\node [protocol, text width=1.6875em, xshift=0.8575em] {}; & \node {Protocol}; \\
	};
\end{tikzpicture}
\caption{An overview of the dependencies of the various resources and protocols discussed in this paper. The arrows point at the resources that are realised by the respective protocol.  Green resources do not require trusted preparations. The functionalities of red resources implement trusted state preparations. The blue resources implement secure delegated computation, and are applications of the protocols in this paper.  Dashed arrows mark trivial constructions that can be obtained by simply restricting the functionality of the more general resource.  The dotted arrow represents a straightforward construction that implies the equivalence of the respective resources, and establishes a link to the storage-and-retrieval model.}
\label{fig:overview}
\end{figure}

\paragraph*{Summary of Contributions.}
The main contributions of this paper are depicted in Figure~\ref{fig:overview} which describes the dependencies of the different resources. We present two independent ways of converting protocols that rely on trusted preparations into protocols that do not need this assumption to be secure.

First, we reduce the remote preparation of single-qubit states from a single plane to a slight alteration of the Remote State Rotation (RSR) resource. While so far RSR was only known to be sufficient for the blind delegation of quantum computations~\cite{MKAC22qenclave}, this new result immediately gives rise to a protocol for the fully secure delegation of $\BQP$ computation that does not require trusted preparations or measurements, and solely relies on trusted single-qubit unitaries by the verifier.
Secondly, we show how the remote preparation of arbitrary single-qubit states can be reduced to the trusted application of multi-qubit Clifford gates and single-qubit unitaries. By the general composability of all involved protocols, we can construct a protocol for the secure delegation of arbitrary quantum computations that does not require trusted preparations or measurements, and relies on trusted gates on a quantum register whose size grows with the security parameter, but is succinct with respect to the size of the computation.

\paragraph*{Limitations.}
All resources in this work that accept a quantum state as an input from the receiver fix its dimension, which is an implicit assumption, or which needs to be guaranteed separately.
For example, if the client is photonic, one would need to assume that it receives single photons and not higher dimensional states.
Otherwise, the client would need to check this condition by relying on non-destructive photon counting, which however -- at this time -- seems experimentally challenging to realise.
It remains an interesting open problem to lift or weaken this assumption of dimensionality through modifications of known protocols.

\paragraph*{Link to Storage-and-retrieval.}
As depicted in Figure~\ref{fig:overview}, the SP-RSP and $\text{RR}_\text{D}$ Resources are equivalent in the sense that a single use of either can be used to securely implement the other in the Abstract Cryptography framework.
This seems to suggest a link to the task of storage-and-retrieval~\cite{SBZ19optimal} of the specific quantum channel implemented by Resource~\ref{res:rrd} ($\text{RR}_\text{D}$). In fact, our result implies that storage-and-retrieval of this channel is possible with probability 1.
Conversely, this seems to be a necessary condition for the construction of a simulator in the security proof of this equivalence, since the simulator is required to succeed in the task of storage-and-retrieval.
We leave for future research the exploration of this link for other resources and quantum channels.

\paragraph*{Application to Quantum Secure Multi-Party Computation.}
The asymmetric QSMPC protocol in~\cite{KKLM23asymmetric} relies on a Collaborative Remote State Preparation protocol for the joint preparation of states in a single plane which guarantees verifiability from the perspective of any single client against arbitrary coalitions of malicious parties. Since it was so far unknown how to design an analogous protocol for the Collaborative Remote Preparation of arbitrary single-qubit states, these results could not be generalised and used for verification protocols requiring states from beyond a single plane. This essentially limited the application of the aforementioned QSMPC protocol to $\BQP$ computations.

Using the results obtained in this paper, the verification of arbitrary quantum computations now only requires trusted applications of single-qubit unitaries and multi-qubit Clifford gates, without relying on trusted preparations or measurements.
These more basic resources can be turned into their collaborative versions using constructions that are very similar to the techniques used in~\cite{KKLM23asymmetric} and~\cite{PLLC23multi}. A formal proof of this claim can be found in Appendix~\ref{app:collaborative_remote_operations}.
This immediately gives rise to an asymmetric QSMPC protocol for arbitrary quantum computations, and with weak clients whose only required quantum abilities are the application of multi-qubit Cliffords and arbitrary single-qubit unitaries.

\paragraph*{Impact on Experiments.}
Ref.~\cite{PLLC23multi} presents an experimental implementation of a multi-client blind delegated computation protocol. Since they allow the client to use single-qubit rotation in a single plane of the Bloch sphere, our techniques could be used directly to upgrade their protocol and experiment to full security by adding a verification procedure at no additional hardware cost or change in the experiment's setup.
This example also shows that our removal of the trust assumption on the quantum state preparation stage opens up the possibility of new architectures or networks for fully secure QSMPC, such as the Qline.

While quantum verification protocols based solely on trusted remote operations as demonstrated in this work seem to rely on two-way quantum communication, this is rather an artefact of the security model in which the states must be communicated between trusted and untrusted parties. In this way, the communication of the states between the source and the remote unitary is made explicit, while it was implicit in previous protocols that assumed the source to be trusted.
However, our proof also captures the case in which the honest party is in possession of the source, but doesn't trust it. In this case -- which might be the most realistic experimentally -- there wouldn't actually be any additional communication compared to the previous protocol with a trusted source, and effectively, the only difference, up to the additional bit-flip, is the removal of the trust assumption on the source. In other terms, the photons must always travel between the source and the hardware which implements the phase rotation, but in one model this communication is implicit and in the other explicit.

\paragraph*{Quantum Verification with Selective NOT.}
In this paper, we showed that it is possible to verify quantum computations with information-theoretic security while only relying on trusted unitary transformations of quantum information, and without trusted preparations of quantum states or their measurements.
While the verification of $\BQP$ computations can be reduced to Resource~\ref{res:rrd} ($\text{RR}_\text{D}$) which requires only single-qubit operations by the verifier, the verification of arbitrary quantum computations, including sampling and those with quantum output, requires not only Resource~\ref{res:ru} (RU), but with Resource~\ref{res:m-rc} (M-RC) also trusted multi-qubit operations, where the number of qubits the verifier operates on simultaneously scales with the security parameter.

In Appendix~\ref{app:non-verif}, we show that Resource~\ref{res:ru} (RU) (and hence trusted single-qubit operations) is indeed sufficient to construct a slightly weaker resource than verifiable RSP, namely Resource~\ref{res:rsp-sn} (RSP-SN) which gives the malicious party some additional power. It remains an open question, however, whether this weaker version of Remote State Preparation is sufficient to securely implement the verification of arbitrary quantum computation (for sampling, or with quantum output), with negligible security error.

\subsection*{Acknowledgments}
The authors would like to thank Theodoros Kapourniotis and Yao Ma for insightful discussions that led to this project.
They also wish to express their gratefulness to Sami Abdul Sater for providing feedback on an early draft of this paper, and pointing out the prior lack of figures.
This work was supported by the European Union’s Horizon 2020 research and innovation program through the FET project PHOQUSING (“PHOtonic Quantum SamplING machine” – Grant Agreement No. 899544).
DL and HO acknowledge funding from the ANR research grant ANR-21-CE47-0014 (SecNISQ).
The work has received funding from the Horizon Europe grant agreement no. 101102140 (Quantum Internet Alliance).
This work has been co-funded by the European Commission as part of the EIC accelerator program under the grant agreement 190188855 for SEPOQC project, and by the Horizon-CL4 program under the grant agreement 101135288 for EPIQUE project.

\phantomsection
\addcontentsline{toc}{section}{References}

\bibliographystyle{splncs04}
\bibliography{../qubib/qubib.bib}

\appendix

\section{Formal Description of the Abstract Cryptography Framework}
\label{app:ac}
\subsection{Protocol Security Model}

In the AC framework, the purpose of a secure protocol $\pi$ is, given a number of available resources $\mathcal{R}$, to construct a new resource -- written as $\pi \mathcal{R}$.  
This new resource can be itself reused in a future protocol. 
A resource~$\mathcal{R}$ is described as a sequence of CPTP maps with an internal state.  
It has \emph{input and output interfaces} describing which party may exchange states with it. 
It works by having each party send it a state (quantum or classical) at one of its input interfaces, applying the specified CPTP map after all input interfaces have been initialised and then outputting the resulting state at its output interfaces in a specified order. 
An interface is said to be \emph{filtered} if it is only accessible by a dishonest player. 
The actions of an honest player $i$ in a given protocol is also represented as a sequence of efficient CPTP maps $\pi_i$ -- called the \emph{converter} of party~$i$ -- acting on their internal and communication registers. We focus here on the two-party setting, in which case $\pi = (\pi_1, \pi_2)$.

In order to define the security of a protocol, we need to give a pseudo-metric on the space of resources. 
We consider for that purpose a special type of converter called a \emph{distinguisher}, whose aim is to discriminate between two resources $\mathcal{R}_1$ and $\mathcal{R}_2$, each having the same number of input and output interfaces. 
It prepares the input, interacts with one of the resources according to its own (possibly adaptive) strategy, and guesses which resource it interacted with by outputting a single bit. 
Two resources are said to be indistinguishable if no distinguisher can guess correctly with good probability.

\begin{definition}[Statistical Indistinguishability of Resources]
  \label{def:ind-res}
  Let $\epsilon > 0$, and let $\mathcal{R}_1$ and $\mathcal{R}_2$ be two resources with same input and output interfaces. 
  The resources are \emph{$\epsilon$-statistically-indistinguishable} if, for all unbounded distinguishers $\mathcal{D}$, we have:
  \begin{equation}
    \label{eq:dist}
    \Bigl\lvert\Pr[b = 1 \mid b \leftarrow \mathcal{D}\mathcal{R}_1] - \Pr[b = 1 \mid b \leftarrow \mathcal{D}\mathcal{R}_2]\Bigr\rvert \leq \epsilon.
  \end{equation}
  We then write $\mathcal{R}_1 \underset{\epsilon}{\approx} \mathcal{R}_2$.
\end{definition}

The construction of a given resource $\mathcal{S}$ by the application of protocol $\pi$ to resource $\mathcal{R}$ can then be expressed as the indistinguishability between resources $\mathcal{S}$ and $\pi \mathcal{R}$. 
More specifically, this captures the correctness of the protocol. 
The security is captured by the fact that the resources remain indistinguishable if we allow some parties to deviate in the sense that they are no longer forced to use the converters defined in the protocol but can use any other CPTP maps instead. 
This is done by removing the converters for those parties in Equation~\ref{eq:dist} while keeping only $\pi_H = \prod_{i \in H} \pi_i$ where $H$ is the set of honest parties. 
The security is formalised as follows in Definition \ref{def:ac-sec} in the case of two parties.

\begin{definition}[Construction of Resources]\label{def:ac-sec}
  Let $\epsilon > 0$.
We say that a two-party protocol $\pi$ $\epsilon$-statistically-constructs resource $\mathcal{S}$ from resource $\mathcal{R}$ if:
\begin{enumerate}
\item It is correct: $\pi \mathcal{R} \underset{\epsilon}{\approx} \mathcal{S}$;
\item It is secure against malicious party $P_i$ for $i \in \{1, 2\}$: there exists a \emph{simulator} (converter) $\sigma_i$ such that $\pi_j\mathcal{R} \underset{\epsilon}{\approx} \mathcal{S} \sigma_i$, where $j \neq i$.
\end{enumerate}
\end{definition}

The General Composition Theorem (Theorem 1 from \cite{MR11abstract-cryptography}) presented below captures the security of protocols which use other secure protocols as subroutines, in sequence or in parallel.
\begin{theorem}[General Composability of Resources]\label{thm:comp-res}
	Let $\mathcal{R}$, $\mathcal{S}$ and $\mathcal{T}$ be resources, $\alpha$, $\beta$ and $\mathsf{id}$ protocols (where protocol $\mathsf{id}$ does not modify the resource it is applied to). Let $\circ$ and $\mid$ denote respectively the sequential and parallel composition of protocols and resources. Then the following implications hold:
	\begin{itemize}
		\item The protocols are \emph{sequentially composable}: if $\alpha \mathcal{R} \!\!\!\underset{\mathit{stat}, \epsilon_{\alpha}}{\approx}\!\!\! \mathcal{S}$ and $\beta \mathcal{S} \!\!\!\underset{\mathit{stat}, \epsilon_{\beta}}{\approx}\!\!\! \mathcal{T}$ then $(\beta \circ \alpha) \mathcal{R} \!\!\!\underset{\mathit{stat}, \epsilon_{\alpha} + \epsilon_{\beta}}{\approx}\!\!\! \mathcal{T}$.
		\item The protocols are \emph{context-insensitive}: if $\alpha \mathcal{R} \!\!\!\underset{\mathit{stat}, \epsilon_{\alpha}}{\approx}\!\!\! \mathcal{S}$ then $(\alpha \mid \mathsf{id}) (\mathcal{R} \mid \mathcal{T}) \!\!\!\underset{\mathit{stat}, \epsilon_{\alpha}}{\approx}\!\!\! (\mathcal{S} \mid \mathcal{T})$.
\end{itemize}
\end{theorem}

Combining the two properties presented above yields concurrent composability (the distinguishing advantage cumulates additively as well). 

\subsection{Security of SDQC Protocols}
\label{subapp:sec-sdqc}

We give here the security properties of the two SDQC constructions that we described in Section \ref{sec:sdqc}. We first introduce the Resource that these protocols construct. It allows a Client to perform a quantum computation without the Server learning anything beyond a controlled leak. The Server cannot modify the output but can force the Client to abort. We use the notation $\cptp C(\rho) = \cptp C \rho \cptp C^\dagger$.

\begin{resource}[ht]
  \caption{Secure Delegated Quantum Computation}
  \label{res:sdqc}
  \begin{algorithmic}[0]
    \STATE \textbf{Inputs:} The Client inputs quantum state $\rho$ and the classical description of a CPTP map $\cptp C$. The Server inputs two bits $(e, c)$. This latter interface is filtered and set to $0$ in the honest case.
    
    \STATE \textbf{Computation by the Resource:}
    \begin{enumerate}
    \item If $e = 1$, the Resource sends the leakage $l_{\rho}$ to the Server's output interface.
    \item If $c = 1$, the Resource outputs $\Abort$ at the Client's output interface.
    \item If $c=0$, it outputs $\cptp C(\rho)$ at the Client's output interface.
    \end{enumerate}
  \end{algorithmic}
\end{resource}

\begin{theorem}[Security of \cite{KKLM23asymmetric}]
\label{thm:sp-sdqc}
The protocol from \cite{KKLM23asymmetric} $\epsilon$-constructs the SDQC Resource \ref{res:sdqc} for $\BQP$ computations from $\abs{V}N$ instances of the Single-Plane Remote State Preparation Resource~\ref{res:sp-rsp}, while only leaking the graph $G = (V, E)$ used in the MBQC computation, with $\epsilon$ negligibly small in the number of repetitions $N$.
\end{theorem}

\begin{theorem}[Security of \cite{KW17garbled}]
\label{thm:full-sdqc}
The protocol from \cite{KW17garbled} $\epsilon'$-constructs the SDQC Resource \ref{res:sdqc} for arbitrary quantum computations from $\order{c(\abs{V} + \abs{E})}$ instances of the Remote State Preparation Resource~\ref{res:rsp} for states $\{\ket{+_{k\pi/4}}\}_{k \in \qty{0, \ldots, 7}} \cup \{\ket{b}\}_{b \in \bin}$, while only leaking the graph $G = (V, E)$ used in the MBQC computation, with $\epsilon'$ negligibly small in the number of correctable errors $c$.
\end{theorem} 
\section{Proof of Theorem \ref{thm:rsp-from-rrd}}
\label{app:rsp-from-rrd}

We prove here the following security result from Section~\ref{sec:rsp-from-rrd}.

\newcounter{tempresult}
\setcounter{tempresult}{\value{theorem}}
\setcounter{theorem}{\value{count:rsp-from-rrd}-1}
\begin{theorem}[Security of Protocol \ref{prot:rsp-from-rrd}]
  Protocol~\ref{prot:rsp-from-rrd} perfectly constructs Resource~\ref{res:sp-rsp} (SP-RSP) from Resource~\ref{res:rrd} ($\text{RR}_\text{D}$).
\end{theorem}
\setcounter{theorem}{\value{tempresult}}

\begin{proof}[Proof of Correctness]
  The state of the output qubit of the honest Receiver by the end of the protocol is:
  \begin{align*}
    \RZ(\theta) \X^b \ket{+} = \ket{+_\theta},
  \end{align*}
  which is exactly the output state of the SP-RSP Resource for input angle $\theta$.
\end{proof}

\begin{proof}[Proof of Soundness]
  To show the security of Protocol~\ref{prot:rsp-from-rrd} against a malicious Receiver, we make use of the explicit construction of Simulator~\ref{sim:rsp-from-rrd}.

  \begin{simulator}[ht]
    \caption{Against a Malicious Receiver}
    \label{sim:rsp-from-rrd}
    \begin{enumerate}	
    \item The Simulator calls the SP-RSP Resource~\ref{res:sp-rsp} and receives a state $\ket{+_\theta}$.
    \item It then emulates the behaviour of the $\text{RR}_\text{D}$ Resource:
      \begin{enumerate}
      \item It receives a single-qubit state $\rho$.
      \item It samples a random bit $b'\sample\bin$ uniformly at random, and performs the following circuit on the two received quantum states:
	
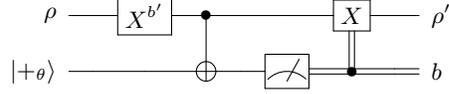
\begin{equation*}
		\Qcircuit @C=1em @R=.7em @!R {
			\lstick{\rho}			& \qw	& \gate{X^{b'}}	& \ctrl{1}	& \qw 	& \qw		& \gate{X}			& \qw	& \rstick{\rho'} \qw \\
			\lstick{\ket{+_\theta}}	& \qw	& \qw			& \targ 		& \qw	& \meter		& \control \cwx \cw	& \cw	& \rstick{b} \cw
		}
	\end{equation*}
      \item The Simulator sends the resulting state $\rho'$ to the Receiver's output interface of the $\text{RR}_\text{D}$ Resource.
      \end{enumerate}
    \end{enumerate}
  \end{simulator}

  In the following, we prove that no Distinguisher can tell apart (i) the ideal world in which the Simulator has single-query oracle access to the SP-RSP Resource~\ref{res:sp-rsp}, and (ii) the real world in which the honest Sender interacts with the $\text{RR}_\text{D}$ Resource~\ref{res:rrd}.
  
  The data that is available to the Distinguisher in both worlds consists of the angle $\theta$ which it inputs at the beginning of the interaction, the quantum state which it inputs at the opposite interface, and the quantum state which the Distinguisher receives at the end. Let $\rho_D$ the internal state of the Distinguisher such that $\rho$ is the state of the subsystem sent as input to the $\text{RR}_\text{D}$ Resource (i.e.~tracing out the internal subsystem). The transcripts in both worlds can be summarised as follows:
  \begin{center}
    \begin{tabular}{cc}\toprule
\textbf{Real World} & \textbf{Ideal World} \\
      \midrule
$\theta$ & $\theta$ \\ 
$\rho_D$ & $\rho_D$ \\
$\cptp U (\rho_D)$ & $\cptp V (\rho_D)$ \\
      \bottomrule
    \end{tabular}
  \end{center}
  In Protocol~\ref{prot:rsp-from-rrd}, we have $\cptp U = \RZ(\theta) X^b \otimes \Id$ where $\Id$ acts on the reference system of the Distinguisher.
  
  In the ideal world, $\cptp V$ is obtained by following the Simulator's circuit. A simple calculation shows that
  \begin{align*}
    \cptp V = \X^b \RZ((-1)^b \theta) \X^{b'} \otimes \Id = \RZ(\theta) \X^{b' \oplus b} \otimes \Id, 
  \end{align*}
  where $\Id$ is applied on the reference system of the Distinguisher. Using the fact that the value of $b'$ is chosen uniformly at random and remains hidden from the Distinguisher, we can substitute it for $b' \oplus b$. The distributions of the transcripts in both scenarios can therefore equivalently described by:
  \begin{center}
    \begin{tabular}{cc}
      \toprule
      \textbf{Real World} & \textbf{Ideal World} \\
      \midrule
      $\theta$ & $\theta$ \\ 
      $\rho_D$ & $\rho_D$ \\
      $\RZ(\theta) \X^b \otimes \Id (\rho_D)$ \qquad & $\RZ(\theta) \X^{b'} \otimes \Id (\rho_D)$  \\
      \bottomrule
    \end{tabular}
  \end{center}
  The two transcripts follow perfectly identical distributions, therefore there is no distinguishing advantage.
\end{proof}

\section{Proof of Theorem \ref{thm:crsp-from-mrc}}
\label{app:crsp-from-mrc}

\setcounter{tempresult}{\value{theorem}}
\setcounter{theorem}{\value{count:crsp-from-mrc}-1}
\begin{theorem}[Security of Protocol \ref{prot:crsp-from-mrc}]
  Let $n$ be the security parameter used in Protocol~\ref{prot:crsp-from-mrc}. Then Protocol~\ref{prot:crsp-from-mrc} $\epsilon_n$-constructs Resource~\ref{res:c-rsp} (C-RSP) from Resource~\ref{res:m-rc} (M-RC), for $\epsilon_n = \frac{1}{2^n - 1}$.
\end{theorem}
\setcounter{theorem}{\value{tempresult}}

\begin{proof}[Proof of Correctness]
Let $\mathbf{0}$ be the all-zero string. The indices $n$ and $n-1$ are used to denote the number of qubits in the state when there is an ambiguity. The state received by the Receiver from the M-RC Resource is equal to:
\begin{align*}
\cptp C_1\ket{\mathbf{0}}_n &= \left(\cptp P_2\cptp C\otimes \X^{\mathbf{r}}\right)\cptp U_g\Z^{\mathbf{d}}\ket{\mathbf{0}}_n\\
&= \cptp P_2\cptp C\ket{0} \otimes \ket{\mathbf{r}}_{n-1}.
\end{align*}
Measuring the last $n-1$ qubits in the computational basis yields $\mathbf{r}$, which is the value expected by the Sender who accepts and sends back the classical description of $\cptp P_2^\dagger$.  After applying this operation, the Receiver has in its possession a single qubit in the state $\cptp C\ket{0}$, which is exactly the output of the C-RSP Resource for input Clifford $\cptp C$.
\end{proof}

\begin{proof}[Proof of Soundness]
To prove the security of Protocol~\ref{prot:crsp-from-mrc} against a malicious Receiver, we make use of Simulator~\ref{sim:crsp-from-mrc}.
  
  \begin{simulator}[ht]
    \caption{Against a Malicious Receiver}
    \label{sim:crsp-from-mrc}
    \begin{enumerate}	
    \item The Simulator calls the C-RSP Resource~\ref{res:rsp} and receives a state $\ket{\phi}$.
    \item It then samples uniformly at random bit-strings $\mathbf{d} \sample \bin^n$ and $\mathbf{r} \sample \bin^{n-1}$. It samples uniformly at random a single-qubit Pauli operator $\cptp Q_2 \sample \mathcal{P}_1$. It samples a random invertible linear function $g$ over $\mathbb{Z}_2^n$.
    \item It then emulates the behaviour of the M-RC Resource as follows:
      \begin{enumerate}
      \item It receives $n$ qubits in state $\rho$.
      \item It applies $\cptp U_g\Z^{\mathbf{d}}$ to these $n$ qubits.
      \item It then discards the first qubit, let $\rho'$ be the resulting state.
      \item It replaces this first qubit with the qubit received from the C-RSP Resource.
      \item It then applies $\cptp Q_2\otimes \X^{\mathbf{r}}$ to these $n$ qubits and outputs them.
      \end{enumerate}
    \item The Simulator then receives $\mathbf{r}' \in \bin^{n-1}$. If $\mathbf{r}' \neq \mathbf{r}$, the Simulator sends $\Abort$. Otherwise it sends the classical description of $\cptp Q_2^\dagger$ and stops.
    \end{enumerate}
  \end{simulator}

  Let $\rho_D$ the internal state of the Distinguisher such that $\rho$ is the state of the subsystem sent as input to the M-RC Resource (i.e.~tracing out the internal subsystem). For notational convenience we consider that the first qubits correspond to the Distinguisher's internal state, while the last $n$ qubits are those sent to either the M-RC Resource or the Simulator.
  
We start by looking at the first two operations that are applied to the state $\rho_D$. These are identical in both worlds and, since $\mathbf{d}$ and $g$ are unknown to the Distinguisher, yield:
\begin{align*}
\widetilde{\rho_D} &= \frac{1}{2^n G_n}\sum_{g, \mathbf{d}} (\Id \otimes \cptp U_g \Z^{\mathbf{d}})(\rho_D)\\ 
&= \frac{1}{G_n}\sum_{g} (\Id \otimes \cptp U_g) \left(\sum_{\mathbf{s} \in \bin^n} p_{\mathbf{s}} \rho_{D, \mathbf{s}}\otimes\ketbra{\mathbf{s}}_n \right)\\
&= p_{\mathbf{0}} \rho_{D, \mathbf{0}}\otimes\ketbra{\mathbf{0}}_n + \frac{1 - p_{\mathbf{0}}}{2^n - 1} \rho_{D, \bar{\mathbf{0}}}\otimes\sum_{\mathbf{s} \neq \mathbf{0}} \ketbra{\mathbf{s}}_n,\end{align*}
where $\rho_{D, \bar{\mathbf{0}}} = \frac{1}{1 - p_{\mathbf{0}}}\sum_{\mathbf{s} \neq \mathbf{0}} p_{\mathbf{s}}\rho_{D, \mathbf{s}}$ and $G_n = \prod_{i = 0}^{n-1} (2^n - 2^i)$ is the number of invertible linear functions over $\mathbb{Z}_2^n$. Summing over the first operation $\Z^{\mathbf{d}}$ performs a partial twirl of the quantum state, which transforms it into a probabilistic mixture of different basis states. The second operation $\cptp U_g$ leaves the state $\ketbra{\mathbf{0}}_n$ unaffected but summing over it perfectly randomises all remaining computational basis states.

The next step differs between the simulation and the real execution. In the real world, the M-RC Resource applies $\cptp P_2 \cptp C$ to the first qubit it received. From the point of view of the Distinguisher, it then receives the state:
\begin{align*}
\rho_{D, r} &= \frac{1}{4} \sum_{\cptp P_2 \in \mathcal{P}_1} (\Id \otimes \cptp P_2 \cptp C \otimes \Id) (\widetilde{\rho_D}) \\
&= \frac{1}{4} \sum_{\cptp P_2 \in \mathcal{P}_1} \left(p_{\mathbf{0}} \rho_{D, \mathbf{0}}\otimes\cptp P_2\cptp C(\ketbra{0}) + \frac{1 - p_{\mathbf{0}}}{2^n - 1} \rho_{D, \bar{\mathbf{0}}}\otimes \cptp P_2\cptp C(\ketbra{1})\right)\otimes\ketbra{\mathbf{0}}_{n-1} \\
&\qquad\qquad\qquad + \frac{2 - 2p_{\mathbf{0}}}{2^n - 1} \rho_{D, \bar{\mathbf{0}}} \otimes \cptp P_2 \cptp C(\Id/2) \otimes \sum_{\mathbf{s}' \neq \mathbf{0}} \ketbra{\mathbf{s}'}_{n-1}.
\end{align*}
On the other hand, the Simulator replaces the first received qubit it received with $\cptp Q_2 \cptp C (\ket{0})$. The Distinguisher then receives:
\begin{align*}
\rho_{D, i} &= \frac{1}{4} \sum_{\cptp Q_2 \in \mathcal{P}_1}\cptp Q_2 \cptp C (\ketbra{0})\otimes\Tr_{1}(\widetilde{\rho_D}) \\
&= \frac{1}{4} \sum_{\cptp Q_2 \in \mathcal{P}_1} \left(p_{\mathbf{0}} \rho_{D, \mathbf{0}} + \frac{1 - p_{\mathbf{0}}}{2^n - 1} \rho_{D, \bar{\mathbf{0}}}\right)\otimes \cptp Q_2 \cptp C (\ketbra{0})\otimes\ketbra{\mathbf{0}}_{n-1} \\
&\qquad\qquad\qquad+ \frac{2 - 2p_{\mathbf{0}}}{2^n - 1} \rho_{D, \bar{\mathbf{0}}} \otimes \cptp Q_2 \cptp C (\ketbra{0}) \otimes \sum_{\mathbf{s}' \neq \mathbf{0}} \ketbra{\mathbf{s}'}_{n-1}
\end{align*}
At this stage, by summing over $\cptp P_2$ or $\cptp Q_2$, we can see that the states $\rho_{D, i}$ and $\rho_{D, r}$ are indistinguishable since:
\begin{align*}
\frac{1}{4} \sum_{\cptp P_2 \in \mathcal{P}_1} \cptp P_2 \cptp C (\ketbra{b}) = \frac{1}{4} \sum_{\cptp P_2 \in \mathcal{P}_1} \cptp P_2 \cptp C (\Id/2) = \frac{1}{4} \sum_{\cptp Q_2 \in \mathcal{P}_1} \cptp Q_2 \cptp C (\ketbra{0}) = \Id/2
\end{align*}

We notice that the states $\rho_{D, r}$ and $\rho_{D, i}$ are both mixtures of states in which the last qubits are always in the computational basis. This means that, up to a unimportant global phase, we have $(\Id\otimes\Z^{\mathbf{d}'})(\rho_{D, r}) = \rho_{D, r}$ and similarly for $\rho_{D, i}$, for any $\mathbf{d}' \in \bin^{n - 1}$. The Simulator and Sender both apply $\X^{\mathbf{r}}$ for a uniformly random $\mathbf{r}$. The state in the hands of the Distinguisher after it receives the $n$ last qubits can therefore be written as:
\begin{align*}
\frac{1}{4^{n-1}} \sum_{\mathbf{r}, \mathbf{d}'} (\Id \otimes \X^{\mathbf{r}}\Z^{\mathbf{d}'})(\rho_{D, r}) = \frac{1}{4^{n-1}} \sum_{\mathbf{r}, \mathbf{d}'} (\Id \otimes \X^{\mathbf{r}}\Z^{\mathbf{d}'})(\rho_{D, i})
\end{align*}
The Distinguisher then applies a unitary $\cptp U$ of its choice to the returned state and its work register and, without loss of generality, measures the last $n-1$ returned qubits in the computational basis. It sends back the measurement outcomes and receives either an abort message or the decryption operation.

Aborting or not only depends on whether the value returned by the Distinguisher $\mathbf{r}'$ is equal to $\mathbf{r}$. The value $\mathbf{r}$ has been sampled according to the same distribution in the real and ideal cases, and the states are identical from the point of view of the Distinguisher at the moment when it is supposed to send the message $\mathbf{r}'$. It is therefore impossible for the Distinguisher to force an abort more often in one scenario since the probability that these two values are equal is the same in both settings. The state in the case of an abort is identical since both the Simulator and the Sender send $\Abort$ and the decryption key is not revealed -- in this case the state of the unmeasured qubit is perfectly mixed in both cases as shown above.

We now analyse the state in the case where there has not been an abort. Up to renormalisation, this corresponds by applying the projector $\dyad{\acc}{\mathbf{0}}_{n-1}\X^{\mathbf{r}}$ to the last $n-1$ qubits. Notice that it is also possible to add the operator $\Z^{\mathbf{d}'}$ to the projector since it is absorbed by $\bra{\mathbf{0}}_{n-1}$. In the real case this yields:
\begin{align*}
\frac{1}{4^{n-1}} \sum_{\mathbf{r}, \mathbf{d}'} \dyad{\acc}{\mathbf{0}}_{n-1} (\Id \otimes \Z^{\mathbf{d}'}\X^{\mathbf{r}})\cptp U (\Id \otimes \X^{\mathbf{r}}\Z^{\mathbf{d}'})(\rho_{D, r})
\end{align*}
The same computation can be performed in the ideal case by replacing $\rho_{D, r}$ with $\rho_{D, i}$. In both settings a Pauli twirl over the last $n-1$ qubits is applied to the operation $\cptp U$ performed by the Distinguisher. The result is that this operation can be seen therefore as a convex combination of Pauli operators on the last $n-1$ qubits tensored with a unitary on the rest of the state:
\begin{align*}
\frac{1}{4^{n-1}} \sum_{\mathbf{r}, \mathbf{d}'} (\Id \otimes \Z^{\mathbf{d}'}\X^{\mathbf{r}})\cptp U (\Id \otimes \X^{\mathbf{r}}\Z^{\mathbf{d}'})(\rho_{D, r}) = \sum_{\cptp P \in \mathcal{P}_{n-1}} p_{\cptp P} \cptp U_{\cptp P}\otimes \cptp P(\rho_{D, r}),
\end{align*}
for a probability distribution $\{p_{\cptp P}\}$ over the $n-1$ qubit Pauli group, and $\cptp U_{\cptp P}$ a unitary acting on the unmeasured qubit and the Distinguisher's internal register. Injecting this expression for the Distinguisher's operation back into the previous equation, we notice that the Pauli can be in fact reduced to simply an $\X^{\mathbf{a}}$ operation, since any $\Z$ component is absorbed by the computational basis measurement. The unitary on the unmeasured qubits can be performed after the measurement and combined with any future operation.

Overall the strategy of the Distinguisher can be expressed by a joint probability distribution over (i) a computational basis input state $\ket{\mathbf{s}}_n$,  (ii) a state for its internal register, (iii) an operation $\X^{\mathbf{a}}$ over $n-1$ qubits and (iv) a unitary $\cptp U$ on the unmeasured qubit and internal state after receiving the correction $\cptp P_2^\dagger$ or $\cptp Q_2^\dagger$ if the result is accepted.

Moreover, this final unitary can without loss of generality be written as $\widetilde{\cptp U} \cptp P^\dagger$, where $\cptp P^\dagger$ is either $\cptp P_2^\dagger$ or $\cptp Q_2^\dagger$. Since $\widetilde{\cptp U}$ acts only on the output state of either the simulation or the real execution and the Distinguisher's internal register, it can be incorporated into the operation used by the Distinguisher to produce its decision bit. Then the internal state is not acted upon by any operation before the final distinguishing decision, we can therefore omit it as well from the description of the state since it does not yield any distinguishing advantage.

In the end, the strategy of the Distinguisher is reduced to simply a joint probability distribution over $\ket{\mathbf{s}}_n$ and $\X^{\mathbf{a}}$. For a fixed value of $\mathbf{a}$, the states in the real and ideal case without the projection $\dyad{\acc}{\mathbf{0}}_{n-1}$ can then be written as:
\begin{align*}
\rho_{\mathbf{a}, r} &= \left(p_{\mathbf{0}} \cptp C(\ketbra{0}) + \frac{1 - p_{\mathbf{0}}}{2^n - 1} \cptp C(\ketbra{1})\right)\otimes\X^{\mathbf{a}}(\ketbra{\mathbf{0}}) + \frac{2 - 2p_{\mathbf{0}}}{2^n - 1} \cptp C(\Id/2) \otimes \sum_{\mathbf{s}' \neq \mathbf{0}} \X^{\mathbf{a}}(\ketbra{\mathbf{s}'})\\
&= \left(p_{\mathbf{0}} \cptp C(\ketbra{0}) + \frac{1 - p_{\mathbf{0}}}{2^n - 1} \cptp C(\ketbra{1})\right)\otimes\ketbra{\mathbf{a}} + \frac{2 - 2p_{\mathbf{0}}}{2^n - 1} \cptp C(\Id/2) \otimes \sum_{\mathbf{s}' \neq \mathbf{a}} \ketbra{\mathbf{s}'}\\
\rho_{\mathbf{a}, i} &= \left(p_{\mathbf{0}} + \frac{1 - p_{\mathbf{0}}}{2^n - 1} \right) \cptp C (\ketbra{0})\otimes\X^{\mathbf{a}}(\ketbra{\mathbf{0}}) + \frac{2 - 2p_{\mathbf{0}}}{2^n - 1} \cptp C (\ketbra{0}) \otimes \sum_{\mathbf{s}' \neq \mathbf{0}} \X^{\mathbf{a}}(\ketbra{\mathbf{s}'})\\
&= \cptp C (\ketbra{0}) \otimes \left(\left(p_{\mathbf{0}} + \frac{1 - p_{\mathbf{0}}}{2^n - 1} \right)\ketbra{\mathbf{a}} + \frac{2 - 2p_{\mathbf{0}}}{2^n - 1} \sum_{\mathbf{s}' \neq \mathbf{a}} \ketbra{\mathbf{s}'}\right)
\end{align*}

If $\mathbf{a} = \mathbf{0}$, then the probability of obtaining the correct measurement outcome is $p_1 = p_{\mathbf{0}}+\frac{1-p_{\mathbf{0}}}{2^n - 1}$ and the post-measurement states are (in the real and ideal case respectively):
\begin{align*}
\rho_{\mathbf{0}, r} &= \frac{1}{p_1}\left(p_{\mathbf{0}} \cptp C(\ketbra{0}) + \frac{1 - p_{\mathbf{0}}}{2^n - 1} \cptp C(\ketbra{1})\right)\otimes\ketbra{\acc}\\
\rho_{\mathbf{0}, i} &= \cptp C (\ketbra{0}) \otimes \ketbra{\acc}
\end{align*}
The trace distance between the post-measurement states in this case is equal to $\delta_1 = \frac{1-p_{\mathbf{0}}}{2^n-2p_{\mathbf{0}}+1}$.

If $\mathbf{a} \neq \mathbf{0}$, then the probability of obtaining the correct measurement outcome is $p_2 = \frac{2 - 2p_{\mathbf{0}}}{2^n - 1}$ and the post-measurement states are (again in the real and ideal case respectively):
\begin{align*}
\rho_{\mathbf{a}, r} &= \cptp C(\Id/2) \otimes\ketbra{\acc}\\
\rho_{\mathbf{a}, i} &= \cptp C (\ketbra{0}) \otimes \ketbra{\acc}
\end{align*}
Here the trace distance is $\delta_2 = 1/2$.

Combining the correct measurement probability and post-measurement distance we get that the overall trace distance is in both cases $p_1\delta_1 = p_2\delta_2 = \frac{1 - p_{\mathbf{0}}}{2^n - 1}$. This is maximised by taking $p_{\mathbf{0}} = 0$, in which case we get that the distinguishing advantage between the real execution and the simulation is $\epsilon_n = \frac{1}{2^n - 1}$, which concludes the proof.
\end{proof}

\section{Proof of Theorem \ref{thm:rsp-from-crsp-and-ru}}
\label{app:rsp-from-crsp-and-ru}

\setcounter{tempresult}{\value{theorem}}
\setcounter{theorem}{\value{count:rsp-from-crsp-and-ru}-1}
\begin{theorem}[Security of Protocol \ref{prot:rsp-from-crsp-and-ru}]
  Protocol~\ref{prot:rsp-from-crsp-and-ru} perfectly constructs Resource~\ref{res:rsp} (RSP) from Resource~\ref{res:c-rsp} (C-RSP) and Resource~\ref{res:ru} (RU).
\end{theorem}
\setcounter{theorem}{\value{tempresult}}

\begin{proof}[Proof of Correctness]
  The single-qubit state which is output by the Receiver at the end of the protocol is:
  \begin{align*}
   \cptp U_2 \cptp U_1 \cptp C \ket{0} =\cptp U \cptp C^\dagger \cptp U_1^\dagger \cptp U_1 \cptp C \ket{0} = \cptp U \ket{0},
  \end{align*}
  which is exactly the output of the RSP Resource for input unitary $\cptp U$.
\end{proof}

\begin{proof}[Proof of Soundness]
  To prove the security of Protocol~\ref{prot:rsp-from-crsp-and-ru} against a malicious Receiver, we make use of Simulator~\ref{sim:rsp-from-crsp-and-ru}.
  
  \begin{simulator}[ht]
    \caption{Against a Malicious Receiver}
    \label{sim:rsp-from-crsp-and-ru}
    \begin{enumerate}	
    \item The Simulator calls the RSP Resource~\ref{res:rsp} and receives a state $\ket{\phi}$.
    \item It then samples a single-qubit unitary $\cptp V_1$ from the Haar measure, and emulates the C-RSP resource by sending $\cptp V_1 \ket{\phi}$.
    \item It then emulates the behaviour of the RU Resource as follows:
      \begin{enumerate}
      \item It receives a single qubit in state $\rho$.
      \item It samples a single-qubit unitary $\cptp V_2$ from the Haar measure, and returns $\cptp V_2 (\rho)$.
      \end{enumerate}
    \item The Simulator then computes $\cptp V_3 = \cptp V_1^\dagger \cptp V_2^\dagger$, and sends the classical description of $\cptp V_3$.
    \end{enumerate}
  \end{simulator}

  Let $\rho_D$ the internal state of the Distinguisher such that $\rho$ is the state of the subsystem sent as input to the RU Resource (i.e.~tracing out the internal subsystem). The transcripts available to the Distinguisher in the real and ideal worlds are:
  \begin{center}
    \begin{tabular}{cc}
      \toprule
      \textbf{Real World} & \textbf{Ideal World} \\
      \midrule
      $\cptp U$ & $\cptp U$ \\ 
      $\cptp C\ket{0}$ & $\cptp V_1 \cptp U \ket{0}$ \\
      $\rho_D$ & $\rho_D$ \\
      $\cptp U_1\otimes\Id (\rho_D)$ & $\cptp V_2\otimes\Id (\rho_D)$ \\
      $\cptp U \cptp C^\dagger \cptp U_1^\dagger$ & $\cptp V_1^\dagger \cptp V_2^\dagger$ \\
      \bottomrule
    \end{tabular}
  \end{center}
  Since $\cptp V_1$ is drawn from the Haar measure, we do not change the distribution of the ideal world transcript by substituting $\cptp V_1 \cptp U$ with a Haar-random $\cptp W_1$, which yields:
  \begin{center}
    \begin{tabular}{cc}
      \toprule
      \textbf{Real World} & \textbf{Ideal World} \\
      \midrule
      $\cptp U$ & $\cptp U$ \\ 
      $\cptp C\ket{0}$ & $\cptp W_1 \ket{0}$ \\
      $\rho_D$ & $\rho_D$ \\
      $\cptp U_1\otimes\Id (\rho_D)$ & $\cptp V_2\otimes\Id (\rho_D)$ \\
      $\cptp U \cptp C^\dagger \cptp U_1^\dagger$ & $\cptp U \cptp W_1^\dagger \cptp V_2^\dagger$ \\
      \bottomrule
    \end{tabular}
  \end{center}
  Applying the unitary operation $\cptp U^\dagger$ -- whose classical description is available to the Distinguisher -- to the last row, we obtain:
  \begin{center}
    \begin{tabular}{cc}
      \toprule
      \textbf{Real World} & \textbf{Ideal World} \\
      \midrule
      $\cptp C\ket{0}$ & $\cptp W_1 \ket{0}$ \\
      $\rho_D$ & $\rho_D$ \\
      $\cptp U_1\otimes\Id (\rho_D)$ & $\cptp V_2\otimes\Id (\rho_D)$ \\
      $\cptp C^\dagger \cptp U_1^\dagger$ & $\cptp W_1^\dagger \cptp V_2^\dagger$ \\
      \bottomrule
    \end{tabular}
  \end{center}
  Note that one can omit the description of $\cptp U$ at this point, as it is entirely independent of all other data contained in the transcripts and therefore yields no distinguishing advantage. Finally, the Distinguisher can apply the unitary operation described by the last row to the quantum state in the third row to obtain:
  \begin{center}
    \begin{tabular}{cc}
      \toprule
      \textbf{Real world} & \textbf{Ideal world} \\
      \midrule
      $\cptp C\ket{0}$ & $\cptp W_1 \ket{0}$ \\
      $\rho_D$ & $\rho_D$ \\
      $\cptp C^\dagger\otimes\Id (\rho_D)$ & $\cptp W_1^\dagger\otimes\Id (\rho_D)$ \\
      $\cptp C^\dagger \cptp U_1^\dagger$ & $\cptp W_1^\dagger \cptp V_2^\dagger$ \\
      \bottomrule
    \end{tabular}
  \end{center}
  As $\cptp U_1$ (in the real world) and $\cptp V_2$ (in the ideal world) are sampled from the Haar measure, the last entries are in fact entirely independent from the rest of the transcripts and identically distributed, and can therefore be removed:
  \begin{center}
    \begin{tabular}{cc}
      \toprule
      \textbf{Real world} & \textbf{Ideal world} \\
      \midrule
      $\cptp C\ket{0}$ & $\cptp W_1 \ket{0}$ \\
      $\rho_D$ & $\rho_D$ \\
      $\cptp C^\dagger\otimes\Id (\rho_D)$ & $\cptp W_1^\dagger\otimes\Id (\rho_D)$ \\
      \bottomrule
    \end{tabular}
  \end{center}
  At this point, distinguishing the two worlds is equivalent to distinguishing an operation drawn uniformly at random from the single-qubit Clifford group ($\cptp C$) from an operation drawn from the single-qubit Haar measure ($\cptp W_1$), when given single-use access to a black-box implementing its inverse and one copy of the quantum state resulting from its application to the $\ket{0}$-state.
  
  We now look more precisely at the structure of state $\rho_D$. In both cases the Distinguisher has in its possession an auxiliary state $\rho_{aux}$, receives either $\cptp C\ket{0}$ or $\cptp W_1 \ket{0}$, applies an operation $\cptp U_D$ of its choice -- which can without loss of generality be considered as a unitary -- to both the received state and its auxiliary state, and sends -- again without loss of generality -- the first qubit to a single-use box which applies either $\cptp C^\dagger$ or $\cptp W_1^\dagger$.
  
  In the real case, denoting $\mathcal{C}_1$ the single-qubit Clifford group, the state that the Distinguisher has in its possession after receiving the output of this box is:
  \begin{align*}
  \frac{1}{24}\sum_{\cptp C \in \mathcal{C}_1} (\cptp C^\dagger\otimes\Id)\cptp U_D(\cptp C\otimes\Id) (\ketbra{0}\otimes\rho_{aux})
  \end{align*}
  Summing over $\cptp C$ performs a Clifford twirl over the action of the unitary $\cptp U_D$ on the first qubit. We can then write:
  \begin{align*}
  \frac{1}{24}\sum_{\cptp C \in \mathcal{C}_1} (\cptp C^\dagger\otimes\Id)\cptp U_D(\cptp C\otimes\Id)(\ketbra{0}\otimes\rho_{aux}) = \sum_{\cptp P \in \mathcal{P}_1} p_{\cptp P} (\cptp P\otimes \cptp U_{\cptp P})(\ketbra{0}\otimes\rho_{aux}),
  \end{align*}
  where  the probabilities satisfy $p_{\X} = p_{\Y} = p_{\Z} = p$ and $p_{\Id} = 1 - 3p$ and $\cptp U_{\cptp P}$ are unitaries acting on the auxiliary state with $\cptp U_\X = \cptp U_\Y = \cptp U_\Z$.
  
  In the ideal case the Distinguisher instead gets:
  \begin{align*}
  \int_{\cptp W_1(d)} (\cptp W_1^\dagger\otimes\Id)\cptp U_D(\cptp W_1\otimes\Id) (\ketbra{0}\otimes\rho_{aux}) d\cptp W_1
  \end{align*}
  We define the equivalence relation $\sim$ on the set of single-qubit unitaries by $\cptp U \sim \cptp V$ if there exists a single-qubit Clifford $\cptp C \in \mathcal{C}_1$ such that $\cptp U = \cptp C \cptp V$. We can then decompose the integral above as:
  \begin{align*}
  \frac{1}{24}\sum_{\cptp C \in \mathcal{C}_1} \int_{\widetilde{\cptp W_1}(d)} (\widetilde{\cptp W_1}^\dagger\otimes\Id)(\cptp C^\dagger\otimes\Id)\cptp U_D (\cptp C\otimes\Id)(\widetilde{\cptp W_1}\otimes\Id) (\ketbra{0}\otimes\rho_{aux}) d\widetilde{\cptp W_1},
  \end{align*}
 where the integral is performed over the equivalence classes of the relation $\sim$. Permuting the sum and integral we get:
  \begin{align*}
  \int_{\widetilde{\cptp W_1}(d)} (\widetilde{\cptp W_1}^\dagger\otimes\Id) \left(\frac{1}{24}\sum_{\cptp C \in \mathcal{C}_1} (\cptp C^\dagger\otimes\Id)\cptp U_D (\cptp C\otimes\Id) \right) (\widetilde{\cptp W_1}\otimes\Id) (\ketbra{0}\otimes\rho_{aux}) d\widetilde{\cptp W_1},
  \end{align*}
  We can then use the Clifford twirl again to obtain:
  \begin{align*}
  \sum_{\cptp P \in \mathcal{P}_1} p_{\cptp P} \int_{\widetilde{\cptp W_1}(d)} (\widetilde{\cptp W_1}^\dagger\otimes\Id)(\cptp P\otimes \cptp U_{\cptp P})(\widetilde{\cptp W_1}\otimes\Id) (\ketbra{0}\otimes\rho_{aux}) d\widetilde{\cptp W_1},
  \end{align*}
  for the same definitions of $p_{\cptp P}$ and $\cptp U_{\cptp P}$. This is a unitary twirl of the depolarising channel over the equivalence classes for relation $\sim$. Applying a unitary $\cptp U$ followed by the depolarising channel and then $\cptp U^\dagger$ has exactly the same effect as simply applying the depolarising channel with the same parameter. Hence the state is equal to:
    \begin{align*}
  \sum_{\cptp P \in \mathcal{P}_1} p_{\cptp P} (\cptp P\otimes \cptp U_{\cptp P})(\ketbra{0}\otimes\rho_{aux}),
  \end{align*}
  which is the same as the real case, thus concluding the proof.
\end{proof}

\section{RSP of Arbitrary Single-qubit States with Selective NOT from Remote Unitary}
\label{app:non-verif}

The ability to prepare arbitrary states on the Bloch sphere is the perfect Resource for SDQC since it allows the Client to prepare not only state in the $\X - \Y$ plane but also states in the computation basis. However, as shown in Section~\ref{sec:arbitrary-states}, this resource is much harder to construct from simple operations if we assume that measurement or state preparation devices cannot be trusted. We present here Resource~\ref{res:rsp-sn} which gives more power to the Adversary in the sense that it can decide to instead prepare the state which is orthogonal to the one desired by the honest User.

\begin{resource}[ht]
  \caption{Remote State Preparation with Selective NOT (RSP-SN)}
  \label{res:rsp-sn}
  \begin{algorithmic}[0]
    \STATE \textbf{Inputs:} The Sender inputs the classical description of a single-qubit unitary $\cptp U$. The Receiver inputs a bit $b \in \bin$. This latter interface is filtered and set to $0$ in the honest case.
    \STATE \textbf{Computation by the Resource:} The Resource prepares and sends a qubit in state $\cptp U\ket{b}$ to the Receiver.
  \end{algorithmic}
\end{resource}

\begin{figure}[t]
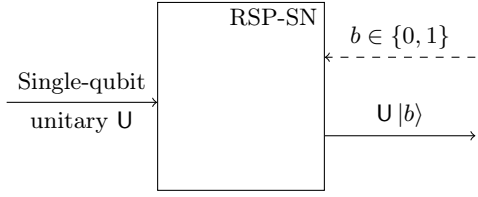

  \centering
  \begin{bbrenv}{rspsn}
    \begin{bbrbox}[name=RSP-SN, minheight=2.5cm]\end{bbrbox}
    \bbrmsgspace{12mm}
    \bbrmsgto{top={Single-qubit},bottom={unitary $\cptp U$},length=2cm}
    \bbrqryspace{6mm}
    \bbrqryfrom{top={$b \in \bin$}, length=2cm, edgestyle={dashed}}
    \bbrqryspace{5mm}
    \bbrqryto{top={$\cptp U\ket{b}$}, length=2cm}
  \end{bbrenv}
  \caption{Ideal Resource~\ref{res:rsp-sn} (RSP-SN). The dashed arrow represents the interface that is accessible only to malicious receivers and that is filtered in the honest case.}
  \label{fig:rsp-sn}
\end{figure}

Protocol~\ref{prot:rsp-sn-from-ru} then shows how to construct the RSP-SN Resource from the Remote Unitary Resource. The Receiver sends a $\ket{0}$ state, the Client applies a random $\Z^{d}$ dephasing and the unitary of its choice and sends it back. This is same randomisation technique as in the previous protocol, only $\Z$ replaces $\X$ since $\ket{0}$ replaces $\ket{+}$.

\begin{protocol}[ht]
  \caption{RSP-SN from RU}
  \label{prot:rsp-sn-from-ru}
  \begin{algorithmic} [0]
    \STATE \textbf{Input:} The Sender inputs the classical description of a single-qubit unitary $\cptp U$.
    \STATE \textbf{Protocol:}
    \begin{enumerate}
    \item The Sender samples a bit $d \sample \bin$ uniformly at random.
    \item The Sender and Receiver call the RU Resource~\ref{res:ru}:
    \begin{itemize}
    	\item The Sender inputs the unitary $\cptp U \Z^d$.
    	\item The Receiver inputs a qubit in the state $\ket{0}$.
    	\item The Resource returns a single qubit to the Receiver who sets it as its output.
    \end{itemize}
    \end{enumerate}
  \end{algorithmic}
\end{protocol}

The security of Protocol~\ref{prot:rsp-sn-from-ru} is given in the following Theorem~\ref{thm:rsp-sn-from-ru}.

\begin{theorem}[Security of Protocol \ref{prot:rsp-sn-from-ru}]
\label{thm:rsp-sn-from-ru}
Protocol~\ref{prot:rsp-sn-from-ru} perfectly constructs Resource~\ref{res:rsp-sn} (RSP-SN) from Resource~\ref{res:ru} (RU).
\end{theorem}

\begin{proof}[Proof of Correctness]
The state received by the Receiver from the RU Resource is equal to:
\begin{align*}
\cptp U \Z^d\ket{0} = \cptp U \ket{0},
\end{align*}
which is exactly the output of the RSP-SN Resource for input unitary $\cptp U$ and honest Receiver.
\end{proof}

\begin{proof}[Proof of Soundness]
To prove the security of Protocol~\ref{prot:rsp-sn-from-ru} against a malicious Receiver, we make use of Simulator~\ref{sim:rsp-sn-from-ru}.
  
  \begin{simulator}[ht]
    \caption{Against a Malicious Receiver}
    \label{sim:rsp-sn-from-ru}
    \begin{enumerate}	
	\item The Simulator receives a qubit in state $\rho$.
	\item It samples a uniformly random bit $d \sample \bin$ and applies $\Z^d$ to the received qubit
	\item It measures the qubit in the computational basis, let $b$ be the outcomes of the measurement.
	\item It sends $b$ to the RSP-SN Resource and receives a qubit in state $\cptp U \ket{b}$.
	\item It sends this qubit to the Receiver's output interface of the RU Resource.
    \end{enumerate}
  \end{simulator}

  Let $\rho_D$ the internal state of the Distinguisher such that $\rho$ is the state of the subsystem sent as input to the RU Resource interface (i.e.~tracing out the internal subsystem). The first $\Z^d$ operation that is applied to the state $\rho_D$ both in the real and ideal settings yields:
\begin{align*}
\widetilde{\rho_D} = \frac{1}{2}\sum_{d \in \bin} (\Z^{\mathbf{d}} \otimes \Id)(\rho_D) = p_0 \ketbra{0}\otimes\rho_{D, 0} + (1 - p_0) \ketbra{1}\otimes\rho_{D, 1}
\end{align*}
where $0 \leq p_0 \leq 1$. Summing over $d$ performs a partial twirl of the quantum state, which transforms it into a probabilistic mixture of the two basis states. We therefore see that the strategy of the Distinguisher boils down to flipping a coin with probability $p_0$ to obtain $0$ and sending $\ketbra{b}\otimes\rho_{D, b}$.

In the real case, the RU Resource then applies $\cptp U$ to the first qubit and the state in the hands of the Distinguisher at the end is $\cptp U (\ketbra{b})\otimes\rho_{D, b}$. In the ideal case, the Simulator recovers the value of $b$ by measuring the state in the computational basis. This does not modify the state since it is already in the state $\ketbra{b}\otimes\rho_{D, b}$. It then sends $b$ to the RSP-SN Resource, receives $\cptp U (\ketbra{b})$ and replaces the qubit received from the Distinguisher by this newly received qubit. The state is then again $\cptp U (\ketbra{b})\otimes\rho_{D, b}$. The real and ideal cases are identically distributed, which concludes the proof.  
\end{proof}

\begin{remark}
Resource \ref{res:rsp-sn} with $n$ qubits can be constructed from an $n$-qubit RU Resource by applying a dephasing $\Z^{\mathbf{d}}$ for $\mathbf{d} \in \bin^n$. The malicious Receiver can in that case choose any basis state $\ket{\mathbf{b}}$ for $\mathbf{b} \in \bin^n$ as input to the filtered interface of the generalised RSP-SN Resource. The security proofs works as above, with perfect security as a result.
\end{remark}

\section{Collaborative Remote Operations}\label{app:collaborative_remote_operations}

The QSMPC protocol in~\cite{KKLM23asymmetric} relies on a protocol for Collaborative Remote State Preparation for states from a single plane.
Similarly, the (non-verifiable) multi-party blind computation protocol in~\cite{PLLC23multi} relies on a protocol for Collaborative Remote State Rotation to generalise from one to many clients.
In this section, we generalise this result from the Remote State Rotation resource to a class of resources that implement remote quantum operations with group structure.
In particular, this class contains all resources without trusted preparations that appear in this paper, thereby showing how to turn the Remote State Rotation with Dephasing ($\text{RR}_\text{D}$), Multi-qubit Remote Clifford (M-RC), and Remote Unitary (RU) resources into their collaborative counterparts.

\begin{resource}[ht]
  \caption{Remote Operation}
  \label{res:ro}
  \begin{algorithmic}[0]
  	\STATE \textbf{Public Information:} A subgroup $\mathfrak{G} \subseteq \operatorname{U}(n)$ of the unitary group of degree $n \in \mathbb{N}$.
    \STATE \textbf{Inputs:} The Sender inputs the classical description of a unitary $\cptp U \in \mathfrak{G}$. The Receiver inputs a quantum state $\rho$ or dimension $n$.
    \STATE \textbf{Computation by the Resource:} The Resource applies $\cptp U$ to the Receiver's input and sends the state $\cptp U(\rho)$ to the Receiver.
  \end{algorithmic}
\end{resource}

\begin{remark}
Resource~\ref{res:ro} (RO) is a generalised version of all other resources in this work that do not rely on trusted preparations.
	Note that:
	\begin{itemize}
		\item Resource~\ref{res:ru} (RU) can be obtained by setting $\mathfrak{G} = \operatorname{U}(2)$.
		\item Resource~\ref{res:m-rc} (M-RC) on $k$ qubits can be obtained by setting $\mathfrak{G} = \mathbf{C}_k \subseteq \operatorname{U}(2^k)$, where $\mathbf{C}_k$ is the Clifford group on $k$ qubits.
		\item Resource~\ref{res:rrd} ($\text{RR}_\text{D}$) can be obtained by setting $\mathfrak{G} = \langle \cptp X, \cptp Z(\pi/4) \rangle$ for $n=2$, and ``forgetting'' the information whether $\cptp X$ was applied on the Sender's side. This is possible because the Sender is always assumed to be honest.
		\item The Remote State Rotation (RSR) Resource from~\cite{MKAC22qenclave,PLLC23multi} can be obtained by setting $\mathfrak{G} = \langle \cptp Z(\pi/4) \rangle$ for $n=2$.
	\end{itemize}
\end{remark}

The goal of the following protocol in then to construct this resource between a purely classical party called the Orchestrator and a quantum Server from $n$ calls to this resource between $n$ potentially malicious Clients and a malicious Server. So long as one client is honest and has access to a trusted RO Resource, the protocol's security guarantees that no malicious coalition of Clients and Server can recover information about the unitary which the Orchestrator wants to apply. Combined together with the results from the other sections, this allows the purely classical Orchestrator to perform RSP with the Server. Replacing the trusted classical Orchestrator with a classical Secure Multi-Party Protocol run by the $n$ Clients allows them to collaboratively generate states from a given ensemble such that no malicious coalition has any information about the state that has been generated.

\begin{protocol}[ht]
  \caption{Collaborative Remote Operation}
  \label{prot:cro}
  \begin{algorithmic} [0]
\STATE \textbf{Inputs:}
    \begin{itemize}
    		\item The $n$ Clients have no input.
    		\item The Orchestrator inputs the classical description of a unitary $\cptp U \in \mathfrak{G}$.
    		\item The Server inputs a quantum state $\rho$.
    \end{itemize}
    \STATE \textbf{Protocol:}
    \begin{enumerate}
    		\item For $j \in \{1,\ldots,n\}$, Client $j$ samples a unitary $\cptp U_j \sample \mathfrak{G}$ from the Haar measure\footnotemark over $\mathfrak{G}$, and sends the classical description of $\cptp U_j$ to the Orchestrator.
      	\item For $j \in \{1,\ldots,n\}$, the Server and Client $j$ use the RO resource, where the client inputs the classical description of $\cptp U_j$ and the Server inputs $\rho_{j-1}$ and obtains $\rho_j = \cptp U_j(\rho_j)$ as output. By convention, $\rho_0 = \rho$.
      	\item The Orchestrator computes $\cptp U' = \cptp U \cptp U_1^\dagger \cdots \cptp U_n^\dagger$ and sends the classical description of $\cptp U'$ to the Server.
       	\item The Server applies the operation $\cptp U'$ to the quantum state that it received from the last call to the RO resource, and keeps the resulting state as its output.
    \end{enumerate}
  \end{algorithmic}
\end{protocol}
\footnotetext{The group of unitaries is unimodular, and so are all of its subgroups. Hence the left and the right Haar measure conveniently coincide.}

\begin{theorem}\label{lemma:cro}
Protocol~\ref{prot:cro} perfectly constructs Resource~\ref{res:ro} (RO) between the Orchestrator and the Server from one use of Resource~\ref{res:ro} (RO) per Client, against arbitrary adversarial patterns of colluding malicious parties, as long as the Orchestrator and at least one Client are honest.
\end{theorem}

\begin{proof}[Proof of Correctness]
    If all participating parties are acting honestly, the quantum state obtained by the Server after the last use of the RO Resource takes the form:
    \begin{align*}
        \cptp U_n \circ \cdots \circ \cptp U_1 ( \rho ).
    \end{align*}
    After the final correction which is applied to this state by the Server, the output becomes:
    \begin{align*}
        \cptp U' \circ \cptp U_n \circ \cdots \circ \cptp U_1 ( \rho ) = \cptp U \circ \cptp U_1^\dagger \circ \cdots \circ \cptp U_n^\dagger \circ \cptp U_n \circ \cdots \circ \cptp U_1 ( \rho) = \cptp U ( \rho ),
    \end{align*}
    which shows that the protocol is correct.
\end{proof}

\begin{proof}[Proof of Soundness]
    Subsequently, we provide the construction of a simulator fit to translate real-world to ideal-world attacks. In the following, we assume the worst case of a colluding malicious coalition of the Server and $n-1$ Clients. Because the protocol and the ideal resource are sufficiently symmetric in the enumeration of the Clients, we can assume without loss of generality that the first Client behaves honestly. The construction of the simulator for this scenario is given as Simulator~\ref{sim:cro}.

  \begin{simulator}[ht]
    \caption{Malicious Server and Clients $2,\ldots,n$}
    \label{sim:cro}
    \begin{enumerate}	
    	\item By impersonating the RO Resource~\ref{res:ro} call between Client $1$ and the Server, it receives a quantum state $\rho$ from the Server.
  		\item The Simulator forwards $\rho$ to the RO Resource~\ref{res:ro} call with the Orchestrator, and receives in return a state $\rho'$.
      	\item It samples the unitary $\cptp U_1 \sample \mathfrak{G}$ from the Haar measure.
      	\item It applies the operation $\cptp U_1$ to $\rho'$ and returns the resulting quantum state to the Server.
      	\item By impersonating the RO Resource call between Client $j$ and the Server, the Simulator receives classical descriptions of unitaries $\cptp U_j \in \mathfrak{G}$ for $j \in \{2,\ldots,n\}$ from the malicious Clients. It receives at each call $j$ a state from the Server and applies $\cptp U_j$ to it before sending it back.
      	\item Finally, the Simulator computes $\cptp  U' = \cptp U_1^\dagger \cdots \cptp U_n^\dagger$ and sends $\cptp U'$ to the Server.
    \end{enumerate}
  \end{simulator}

    It remains to be shown that the views of the Distinguisher in the real world where it has access to the inputs $\cptp U, \rho$ and to the views of all malicious parties, and in the ideal world where it has access to the inputs $\cptp U, \rho$ and to the interfaces to the Simulator are perfectly equal. These two views can be summarised as follows:

\begin{center}
\begin{tabular}{lcc}
\hline
& \textbf{Real World}\hspace*{1em} & \textbf{Ideal World} \\
\hline
\textit{Input unitary} & $\cptp U$ & $\cptp U$ \\
\textit{Input state} & $\rho$ & $\rho$ \\
\textit{Client $1$ Output} & $\cptp U_1 (\rho)$ & $\cptp U_1 \circ \cptp U (\rho)$ \\
\textit{Correction} & $\cptp U \cptp U_1^\dagger \cdots \cptp U_n^\dagger$ & $\cptp U_1^\dagger \cdots \cptp U_n^\dagger$ \\
\hline
\end{tabular}
\end{center}

    Since $\cptp U_1$ is sampled from the Haar measure by the Simulator in the ideal world, by the invariance of the Haar measure, we can substitute it by $\cptp U_1 \cptp U^\dagger$ without changing the view of the Distinguisher. This yields:

\begin{center}
\begin{tabular}{cc}
\hline
\textbf{Real world}\hspace*{1em} & \textbf{Ideal world} \\
\hline
$\cptp U$ & $\cptp U$ \\
$\rho$ & $\rho$ \\ 
$\cptp U_1 (\rho)$ & $\cptp U_1 (\rho)$ \\
$\cptp U \cptp U_1^\dagger \cdots \cptp U_n^\dagger$ & $\cptp U \cptp U_1^\dagger \cdots \cptp U_n^\dagger$ \\
\hline
\end{tabular}
\end{center}

    Clearly, these two distributions are identical. Furthermore, the Simulator acts exactly as the RO Resource when called between malicious Client $j$ and the Server, which together with the views above proves that the views of the Distinguisher in the two worlds are perfectly indistinguishable.
\end{proof}

\end{document}